\documentclass[journal]{IEEEtran}

\usepackage{verbatim}
\usepackage{CJK}
\usepackage{indentfirst}
\usepackage{amsmath}
\usepackage{amsthm}
\usepackage{graphicx}
\usepackage{subfigure,color}
\usepackage[dvipsnames]{xcolor}
\usepackage{booktabs}
\usepackage{multirow}
\usepackage{tabularx}
\usepackage{makecell}

\ifCLASSINFOpdf
\else
\fi

\begin{document}
%
\title{Open Data in the Digital Economy: An Evolutionary Game Theory Perspective}
%
%
%

\author{Qin Li,
        Bin Pi,~\IEEEmembership{Student Member,~IEEE,}
        Minyu Feng,~\IEEEmembership{Member,~IEEE,} and J\"{u}rgen Kurths 
\thanks{This work was supported in part by the National Nature Science Foundation of China (NSFC) under Grant No. 62206230, in part by the Humanities and Social Science Fund of Ministry of Education of the People's Republic of China under Grant No. 21YJCZH028, and in part by the Natural Science Foundation of Chongqing under Grant No. CSTB2023NSCQ-MSX0064.}
\thanks{Qin Li is with the School of Public Policy and Administration, Chongqing University, Chongqing 400044, P. R. China.}
\thanks{Bin Pi is with the School of Mathematical Sciences, University of Electronic Science and Technology of China, Chengdu 611731, P. R. China.}
\thanks{Minyu Feng is with the College of Artificial Intelligence, Southwest University, Chongqing 400715, P. R. China. (e-mail: myfeng@swu.edu.cn.)}
\thanks{J\"{u}rgen Kurths is with the Potsdam Institute for Climate Impact Research, 14437 Potsdam, Germany, and also with the Institute of Physics, Humboldt University of Berlin, 12489 Berlin, Germany.}}

%
%

\markboth{IEEE TRANSACTIONS ON COMPUTATIONAL SOCIAL SYSTEMS}%
{Shell \MakeLowercase{\textit{et al.}}: Bare Demo of IEEEtran.cls for IEEE Journals}
%



\maketitle

\begin{abstract}
Open data, as an essential element in the sustainable development of the digital economy, is highly valued by many relevant sectors in the implementation process. However, most studies suppose that there are only data providers and users in the open data process and ignore the existence of data regulators. In order to establish long-term green supply relationships between multi-stakeholders, we hereby introduce data regulators and propose an evolutionary game model to observe the cooperation tendency of multi-stakeholders (data providers, users, and regulators). The newly proposed game model enables us to intensively study the trading behavior which can be realized as strategies and payoff functions of the data providers, users, and regulators. Besides, a replicator dynamic system is built to study evolutionary stable strategies of multi-stakeholders. In simulations, we investigate the evolution of the cooperation ratio as time progresses under different parameters, which is proved to be in agreement with our theoretical analysis. Furthermore, we explore the influence of the cost of data users to acquire data, the value of open data, the reward (penalty) from the regulators, and the data mining capability of data users to group strategies and uncover some regular patterns. Some meaningful results are also obtained through simulations, which can guide stakeholders to make better decisions in the future.
\end{abstract}


%
\IEEEpeerreviewmaketitle

\section{Introduction}

\IEEEPARstart{W}{ith} the development of digitization and digital economy in current societies \cite{khalin2018digitalization}, the requirement for accessing the open data has increased, which in turn has driven data providers (e.g., governments, organizations, and enterprises) to make their data available over the internet and allow data users to reuse these data in various applications and online services \cite{goldfarb2019digital, jetzek2014data}. This kind of data supply and utilization is called open data initiatives, defined as a concept for enabling free access, using, and sharing of open data for the general public without any limitations \cite{neves2020impacts}. Open data can be utilized to analyze citizen participation and stakeholder decision-making \cite{cantador2020exploiting}. By using open data, governments, businesses, and individuals can create value and bring about social \cite{davies2016researching}, economic \cite{chetty2020economic}, and environmental benefits \cite{shang2021impacts}. Therefore, it is necessary to study open data for the development of society.

In response to this trend, government agencies and institutions around the world in recent years have been increasingly active in adopting and implementing open data initiatives \cite{altayar2018motivations}. However, there are still several challenges to achieving the full potential of open data and supporting all interested parties with the publication and consumption of this data \cite{dlamini2019data}. The open data initiatives abound with social, political, and economic interactions that their actions affect each other. The extensive literature on open data, which includes multi-stakeholders, such as the public sector that produces, collects, maintains, and disseminates large amounts of data, the business that needs the data to advertise their economic activities, and the public as potential users and providers, also concludes that different stakeholders have different requirements and influences on the interests of open data initiatives \cite{saxena2017usage, gonzalez2015multiple}. Moreover, the adoption of open data from different stakeholders is an extremely dynamic process involving multiple actions of stakeholders and changes as time progresses \cite{safarov2017utilization}. However, previous research on data providers has mainly focused on the empirical analysis from a binary logic that regards data as open or non-open dominating, while ``open'' and ``non-open'' are not absolutes in practice \cite{raman2012collecting}. This lacks an understanding of the interactions between different stakeholders, and strategic trends among them. Understanding these factors is necessary to facilitate the success of open data. In addition, most previous studies on open data have considered only data providers and users, while neglecting the existence of data regulators. The data regulator \cite{eichler2012open} plays an important role in open data by monitoring the data offered by data providers to data users and rewarding or punishing their behavior appropriately, which can lead to the development of open data in a good direction, i.e., data providers provide high quality data and data users actively acquire data \cite{zhang2020customer}.

Due to the rapid development of evolutionary game theory in recent years, it has been widely employed in studying and predicting evolutionary behavior in biology \cite{mcnamara2020game}, economics \cite{kabir2020evolutionary}, management \cite{tian2019evaluating, qu2022evolutionary}, social behaviors \cite{pi2022evolutionary, pi2022evolutionary2} and so on \cite{li2021three, li2020perception}. Besides, many recent reviews \cite{jusup2022social, capraro2021mathematical} have pointed out the fact that the study of evolutionary game theory has many applications across the social and natural sciences. Therefore, we utilize evolutionary game theory to provide a theoretical framework for our analysis of the open data. The most important approach focusing on the evolutionary game dynamics is the replicator, which was first defined for single species by Taylor and Jonker \cite{taylor1978evolutionary}, and then has been widely studied in \cite{wang2020evolutionary, li2021three, fan2020evolutionary}. Therefore, we investigate the open data based on the idea of evolutionary game theory in this paper. Concretely, considering bounded rationality and learning mechanisms, this study adopts the method from evolutionary game theory as a theoretical perspective to provide explanations and understanding of how open data development is embedded into different stakeholder strategic interactions. Under the context of open data, our research proposes an analytical framework that describes the dynamic interaction between data providers, users, and regulators, where data providers face strategic options (open high quality data and open low quality data), data users have an either-or strategy (acquire data or do not), and data regulators also have two alternative strategies (regulate open data and do not). In particular, we attempt to answer the following questions: How will the strategies of the three multi-stakeholders evolve under different conditions? What can be done to promote good open data development, i.e., data providers choose to open high quality data, and data users choose to acquire data with a positive attitude?

In particular, the primary contributions of this paper can be succinctly summarized as follows:

\begin{itemize}
\item We utilize evolutionary game theory to depict the dynamic interaction among data providers, users, and regulators within the realm of open data.
\item We introduce the data mining capabilities of data users and take the presence of data regulators into account, making our approach more aligned with real-world scenarios.
\item The evolutionary stable strategies of the three parties have been illustrated not only from mathematical standpoints but also through simulations.
\item The impacts of the parameters in the proposed evolutionary game model on group strategies are investigated from various perspectives.
\item Recommendations are made to support the sustainable development of open data based on the theoretical analysis and numerical simulation results.
\end{itemize}

The remaining of this paper is organized as follows. In Section \ref{Model}, we present a game model to study the open data, and two theorems are proven for analyzing the evolutionary stable strategy (ESS) based on the payoff matrix and the replicator dynamic system. After that, the simulations are conducted to verify the ESS and investigate the model parameters on the group behaviors in Section \ref{Simulation}. Finally, we summarize the conclusions for the implementation of open data and describe an outlook for future research in Section \ref{Conclusion}.

\section{An evolutionary game model of stakeholders in open data} \label{Model}

Exploring the game of open data can unlock the potential value of data and promote the development of the social economy. Therefore, effective modeling of the open data process is crucial. Most previous studies assumed that there are only data providers and users in the open data process, and ignored the fact that there are also data regulators. In this paper, we add data regulators to the game model and introduce the data mining capability of data users, which is substantially different from previous works. In this section, we mainly describe the calculation of the three types of stakeholders' payoffs based on reasonable assumptions and then establish the replicator dynamic system. Subsequently, we discuss the conditions under which the equilibrium point becomes the ESS in order to suggest decisions for sustainable open data initiatives.

\subsection{Model Description}

Advocating for open high quality data is a long-term process that requires data providers to continuously and steadily provide high quality data to society, which has attracted commentators, strategists, and practitioners from a variety of origins bringing different understandings and visions. In terms of the data providers, they are unlikely to be completely consistent in taking open high quality data actions due to the different data resources and policy goals, and their costs, as well as benefits, are calculated based on their actions. For example, for conservative data providers, they provide low quality data for fear of getting less reward than they pay for it. Additionally, after repeated learning and accumulation, data providers might adjust and improve their strategies according to the feedback in the early experience as time progresses. Besides, data users have different levels of need for open data. To gain an advantage in the digital market, data users take measures to utilize open data and create value for themselves. However, for data users with weak data mining capabilities and small data needs, they are more likely to take the strategy of not acquiring open data. Moreover, the data regulators can regulate the strategy between data providers and users on open data, and their strategies will be influenced by their own payoffs, which can essentially be attributed to the influence of the strategy of the three types of stakeholders.

Hereby, we utilize an evolutionary game theoretical framework in which the data providers, users, and regulators are the players in the game and the goal of the players is to maximize their own payoffs. In the process of the game, the three players obtain different payoffs based on different strategy interactions. On the basis of evolutionary game theory (EGT), the three players play the game under bounded rationality, and their decisions are independent and made from the perspective of information asymmetry, i.e., they do not know each other's strategies and their strategic choices are influenced by the results of the previous stage of the game. Concretely, the strategy choices of the data providers, users, and regulators are given below:

(A) data providers: usually have two strategies when making decisions:
\begin{description}
  \item[(i)] One is to open high quality data (we denote the strategy as cooperation, denoted by the symbol Cp), and the fraction is $x$;
  \item[(ii)] While the other is to open low quality data (we denote the strategy as defection, denoted by the symbol Dp), and the fraction is $1 - x$ for normalization.
\end{description}

(B) data users: also have two strategies when making decisions:
\begin{description}
  \item[(i)] One is to acquire data (we denote the strategy as cooperation, denoted by the symbol Cu), and the fraction is $y$;
  \item[(ii)] While the other is not to acquire data (we denote the strategy as defection, denoted by the symbol Du), and the fraction is $1 - y$ for normalization.
 \end{description}
  
 (C) data regulators: also have two strategies when making decisions:
\begin{description}
  \item[(i)] One is to regulate open data (we denote the strategy as cooperation, denoted by the symbol Cr), and the fraction is $z$;
  \item[(ii)] While the other is not to regulate open data (we denote the strategy as defection, denoted by the symbol Dr), and the fraction is $1 - z$ for normalization.
\end{description}

Among the data providers, some adopt the Cp strategy, while others adopt the Dp strategy. Analogously, for the data users and data regulators, some of them use the Cu and Cr strategy, while others utilize the Du and Dr strategy. Therefore, the proportion $x$ of data providers utilizing strategy Cp, the proportion $y$ of data users adopting strategy Cu, and the proportion $z$ of data regulators using strategy Cr follow $($x$, $y$, $z$) \in [0, 1] \times [0, 1] \times [0, 1]$. When the payoffs of the three players are lower than their respective average payoffs, rational players (data providers or data users or data regulators) will adjust their strategy during the game process to increase their payoffs, which yields $x$, $y$, and $z$ changing with time, namely, the proportion of the population choosing different strategies changes dynamically. Thus, we propose an EGT model to investigate the competition and cooperation among data providers, users, and regulators and further study the ESS of the three players. Fig. \ref{game_relationship} depicts the game relationship among data providers, users, and regulators about open data.

\begin{figure}[ht]
\centering
\includegraphics[scale=0.40]{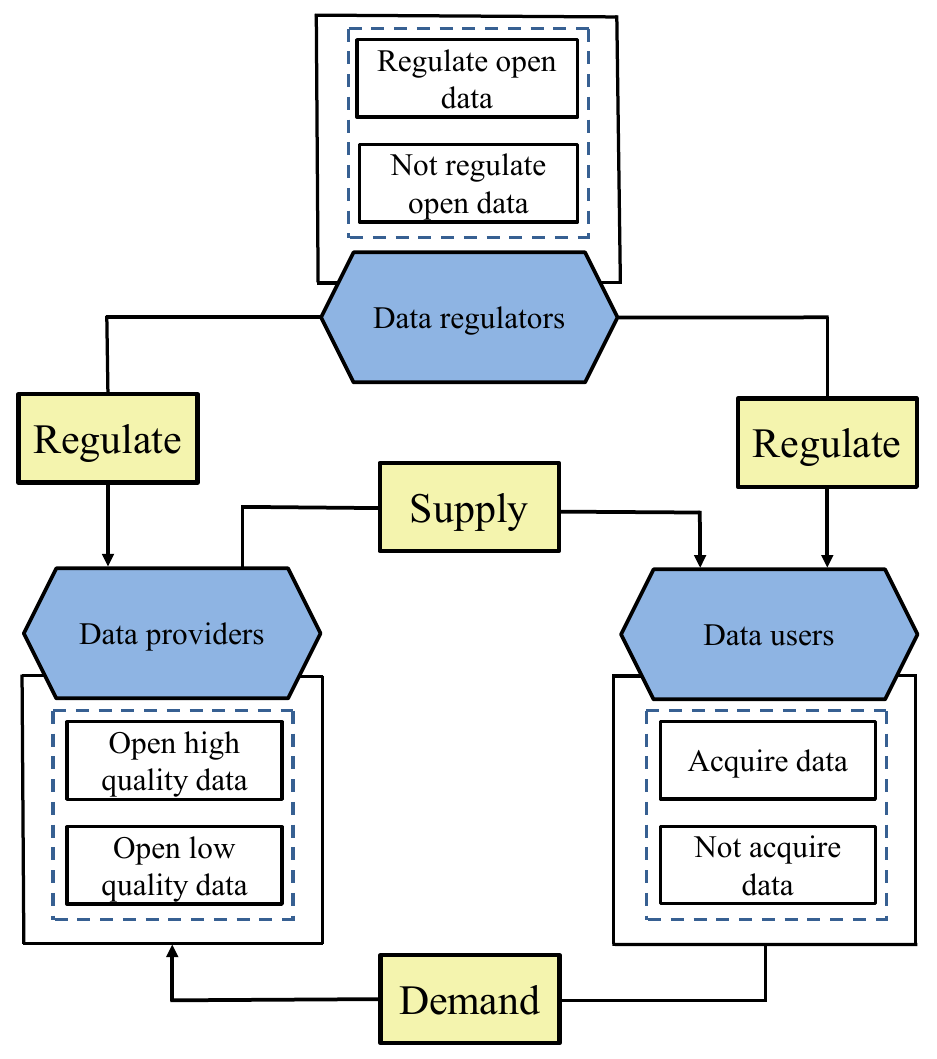}
\caption{\textbf{The game relationship among data providers, users, and regulators.} Each of them has two different strategies to choose from, and the combination of strategies among them will have an impact on their payoffs.}
\label{game_relationship}
\end{figure}

\subsection{Assumptions and Notations}

Herein, we make the assumptions and notations to facilitate modeling. We propose a three-party evolutionary game model in which the data providers, users, and regulators are each infinite populations, respectively. As aforementioned pointed out, there are two strategies for them. The data provider is supposed to consider the cost of sharing data in terms of labor, time, and the potential loss of privacy, weigh the benefits and losses and choose to adopt Cp or Dp strategy. We suppose that the cost of sharing high quality data by the data provider is $C_1$ and 0 for sharing low quality data since no particular collation is required. The value of high quality data is $V_1$, while that of low quality data is $V_2$. If the data providers provide high quality data that are acquired by data users, the satisfaction of data users and the public to the data provider will be enhanced and the data provider's credibility will be strengthened. Then, in the long run, the data providers and users share synergistic benefits $R_1$. On the contrary, if the low quality data provided by data providers are positively acquired by data users, the credibility of the data provider will be reduced and the payoff of the data provider will be $-L_1$. For the purpose of promoting the emergence of good phenomena, where data providers provide high quality data and data users acquire data, the data providers are rewarded (penalized) with a certain amount of $F$ for providing high (low) quality data if they are under strict regulation by the data regulators.

For the data user, the cost of acquiring the open data (including but not limited to collecting and collating the data) is $C_2$, releasing the potential economic value of the data, while different data users have different data mining capabilities (denote it as $\alpha$). However, if the data users do not acquire high quality data, while other competing data users acquire it and exploit the potential value, it leads to a payoff $-L_2$ to the data user in the long term. Analogously, the data users are rewarded (penalized) with a certain amount of $F$ for acquiring (not acquiring) data if they are under strict regulation by the data regulator to promote the emergence of good phenomena.

For the data regulator, the cost of performing the regulation is $C_3$, and the income for regulation is $P$. If a good phenomenon emerges under strict regulation by the data regulator, the data regulator receives the reward $R_2$.

The notations utilized in this model are briefly summarized in Tab. \ref{notations}.

\begin{center}
\begin{table*}
\centering
\caption{Notations of the game model}
\label{notations}
\renewcommand\arraystretch{1.2}
\begin{tabular}{cc}
\Xhline{1.5pt}
Notations & Descriptions \\
\Xhline{0.5pt}
$x$&The fraction of the data providers choosing to open high quality data\\
$y$&The fraction of the data users choosing to acquire data\\
$z$&The fraction of the data regulators choosing to regulate open data\\
$C_1$&The cost of the data providers providing high quality data\\
$C_2$&The cost of the data users acquiring data\\
$C_3$&The cost of the data regulators regulating open data\\
$V_1$&The value of the high quality data\\
$V_2$&The value of the low quality data\\
$R_1$&The synergistic benefits of data providers and users when high quality data is acquired by data users\\
$R_2$&The reward for data regulators performing regulation when data providers provide high quality data and data users acquire data\\
$L_1$&The loss of data providers when low quality data is acquired by data users\\
$L_2$&The loss of data users when high quality data is acquired by competitors\\
$P$&The income for data regulators choosing to regulate open data\\
$F$&The reward (penalty) for data providers and users by data regulators under the regulation of open data\\
$\alpha$&The data mining capability of data users\\
\Xhline{1.5pt}
\end{tabular}
\end{table*}
\end{center}

\subsection{The Calculation of the Three Players' Payoffs}

\begin{table*}[htbp]
\centering
\caption{Payoff matrix}
\label{matrix}
\begin{tabular}{|cc|c|cc|}
\hline
\multicolumn{2}{|c|}{\multirow{2}{*}{}}                                               & \multirow{2}{*}{data provider}                                           & \multicolumn{2}{c|}{data user}                                                                                                           \\ \cline{4-5} 
\multicolumn{2}{|c|}{}                                                                &                                                                          & \multicolumn{1}{c|}{\begin{tabular}[c]{@{}c@{}}acquire\\ (y)\end{tabular}} & \begin{tabular}[c]{@{}c@{}}not acquire\\ (1-y)\end{tabular} \\ \hline
\multicolumn{1}{|c|}{\multirow{4}{*}{data regulator}} & \multirow{2}{*}{\begin{tabular}[c]{@{}c@{}}regulate\\ (z)\end{tabular}}    & \begin{tabular}[c]{@{}c@{}}provide high quality data\\ (x)\end{tabular}  & \multicolumn{1}{c|}{$P+R_2-C_3,R_1+F-C_1,R_1+F+\alpha V_1-C_2$}                                                 & $P-C_3,F-C_1,-F-L_2$                                                     \\ \cline{3-5} 
\multicolumn{1}{|c|}{}                                &                               & \begin{tabular}[c]{@{}c@{}}provide low quality data\\ (1-x)\end{tabular} & \multicolumn{1}{c|}{$P-C_3,-F-L_1,F+\alpha V_2-C_2$}                                                     & $P-C_3,-F,-F$                                                           \\ \cline{2-5} 
\multicolumn{1}{|c|}{}                                & \multirow{2}{*}{\begin{tabular}[c]{@{}c@{}}not regulate\\ (1-z)\end{tabular}} & \begin{tabular}[c]{@{}c@{}}provide high quality data\\ (x)\end{tabular}  & \multicolumn{1}{c|}{$0,R_1-C_1,R_1+\alpha V_1-C_2$}                                                     & $0,-C_1,-L_2$                                                           \\ \cline{3-5} 
\multicolumn{1}{|c|}{}                                &                               & \begin{tabular}[c]{@{}c@{}}provide low quality data\\ (1-x)\end{tabular} & \multicolumn{1}{c|}{$0,-L_1,\alpha V_2-C_2$}                                                     & $0,0,0$                                                           \\ \hline
\end{tabular}
\end{table*}

As mentioned above, there are three players (data providers, users, and regulators) in the player set, and different players choose from different strategy spaces. Based on the aforementioned assumptions and notations, the payoffs for players depend on the strategy choices of them, which are illustrated in Tab. \ref{matrix}. Let $E_{p1}$ and $E_{p2}$ represent the expected payoffs of the data providers who provide high quality data and low quality data, respectively. According to the payoff matrix in Tab. \ref{matrix}, $E_{p1}$ and $E_{p2}$ consist of four parts. We take $E_{p1}$ as an example, $E_{p2}$ is similar. The four parts of $E_{p1}$ are:

\begin{description}
  \item[(i)] The probability that the data user chooses to acquire a strategy is multiplied by the probability that the data regulator selects to regulate strategy, and then multiplied by the payoff when the data providers perform to provide high quality data strategy;
  \item[(ii)] The probability of the data user selecting not to acquire a strategy times the probability of the data regulator choosing to regulate strategy, and then times the payoff when the data providers adopt to provide high quality data strategy;
  \item[(iii)] The probability that the data user performs to acquire a strategy is multiplied by the probability that the data regulator chooses not to regulate strategy, and then multiplied by the payoff when the data providers select to provide high quality data strategy;
  \item[(iv)] The probability of the data user adopting not to acquire a strategy multiplies the probability of the data regulator selecting not to regulate strategy and then multiplies the payoff when the data providers perform to provide high quality data strategy.
\end{description}

Therefore, according to the payoff matrix shown in Tab. \ref{matrix}, we yield

\begin{equation}
\begin{aligned}
E_{p1}&=yz(R_1+F-C_1)+(1-y)z(F-C_1)
\\&+y(1-z)(R_1-C_1)+(1-y)(1-z)(-C_1)
\\&=zF+yR_1-C_1.
\end{aligned}
\end{equation}

Analogously, we obtain the expected payoffs $E_{p2}$ of the data providers providing low quality data, which can be expressed as

\begin{equation}
\begin{aligned}
E_{p2}&=yz(-F-L_1)+(1-y)z(-F)+y(1-z)(-L_1)
\\&=-zF-yL_1.
\end{aligned}
\end{equation}

According to the payoff matrix, the expected payoff $E_{u1}$ and $E_{u2}$ for the data user to select acquire strategy and not acquire strategy and average expected payoff could be denoted as in the following equations, where the calculation process is similar to $E_{p1}$ and $E_{p2}$,

\begin{equation}
\begin{small}
\begin{aligned}
E_{u1}&=xz(R_1+F+\alpha V_1-C_2)+(1-x)z(F+\alpha V_2-C_2)
\\&+x(1-z)(R_1+\alpha V_1-C_2)+(1-x)(1-z)(\alpha V_2-C_2)
\\&=zF+xR_1+x\alpha (V_1-V_2)+\alpha V_2-C_2
\end{aligned}
\end{small}
\end{equation}
and
\begin{equation}
\begin{aligned}
E_{u2}&=xz(-F-L_2)+(1-x)z(-F)+x(1-z)(-L_2)
\\&=-zF-xL_2.
\end{aligned}
\end{equation}

Besides, we utilize $E_{r1}$ and $E_{r2}$ to denote the expected payoff of the data regulator selecting the regulate strategy and not regulate strategy, respectively. Similar to the calculation process for $E_{p1}$ and $E_{p2}$, we have

\begin{equation}
\begin{aligned}
E_{r1}&=xy(P+R_2-C_3)+(1-x)y(P-C_3)
\\&+x(1-y)(P-C_3)+(1-x)(1-y)(P-C_3)
\\&=xyR_2+P-C_3
\end{aligned}
\end{equation}
and
\begin{equation}
E_{r2}=0.
\end{equation}

\subsection{The Replicator Dynamic System}

Based on the aforementioned payoff analysis, the average payoffs of the data providers ($\overline{E_p}$), users ($\overline{E_u}$), and regulators ($\overline{E_r}$) can be denoted as $\overline{E_p} = xE_{p1}+(1-x)E_{p2}$, $\overline{E_u} = yE_{u1}+(1-y)E_{u2}$, and $\overline{E_r} = zE_{r1}+(1-z)E_{r2}$, respectively.

As we stated before, the proportion $x$ of data providers adopting strategy Cp, the proportion $y$ of data users utilizing strategy Cu, and the proportion $z$ of data regulators using strategy Cr will change as time progresses, i.e., $x$, $y$, and $z$ are all functions of time $t$. According to the work by Friedman (1991) \cite{friedman1991evolutionary} and Taylor and Jonker (1978) \cite{taylor1978evolutionary}, the percentage of players grows, when the payoff of those players is greater than the average payoff of the entire population and the growth rate is represented by differential equations in continuous time. Besides, the growth rate of a strategy selected by the players should be equal to its payoff less the population average payoff among each player. Therefore, the replicator dynamic equations of opening high quality data selected by the data providers ($F_x$), acquiring data selected by the data users ($F_y$), and regulating open data selected by the data regulators ($F_z$) are as follows:

\begin{equation}\label{F_x}
\begin{cases}
\begin{aligned}
F_x&=\frac{dx}{dt}=x(E_{p1}-\overline{E_p})=x(1-x)(E_{p1}-E_{p2})\\&=x(1-x)(2zF+yR_1+yL_1-C_1)\\
F_y&=\frac{dy}{dt} =y(E_{u1}-\overline{E_u})=y(1-y)(E_{u1}-E_{u2})\\&=y(1-y)(2zF+xR_1+xL_2+x\alpha (V_1-V_2)\\&+\alpha V_2-C_2)\\
F_z&=\frac{dz}{dt} =z(E_{r1}-\overline{E_r})=z(1-z)(E_{r1}-E_{r2})\\&=z(1-z)(xyR_2+P-C_3).\\
\end{aligned}
\end{cases}
\end{equation}

Eq. \ref{F_x} is the continuous frequency dynamic system for the data provider, user, and regulator population. From the above derivation, we obtain the replicator dynamic equations of the three populations, which demonstrate that the proportion of a certain strategy adopted by the game player is a time-varying variable. Besides, the number and the positive or negative sides of the equations reflect the speed and direction of changes in the proportion of population strategies as time progresses.

\subsection{Model Analysis and Discussion}\label{Model discussion}

Now, we discuss the evolutionary stable strategy (ESS) of the replicator dynamic system. From Eq. \ref{F_x}, two theorems can be derived as follows, one of them provides the equilibrium points of the replicator dynamic system, while the other gives the conditions for the equilibrium point to become an ESS.

\newtheorem{thm}{Theorem}

\begin{thm}\label{equilibrium points}
For the replicator dynamic system represented by Eq. \ref{F_x}, the equilibrium points of the system are eight fixed points (0, 0, 0), (0, 0, 1), (0, 1, 0), (0, 1, 1), (1, 0, 0), (1, 0, 1), (1, 1, 0), (1, 1, 1) and one mixed point ($x^*$, $y^*$, $z^*$), which satisfies $F_{x^*}=0$, $F_{y^*}=0$, and $F_{z^*}=0$, respectively.
\end{thm}

\begin{proof}
The fixed eight points are obvious. According to the stability theory of differential equations, a point of the replicator dynamic system represented by Eq. \ref{F_x} is an equilibrium point if it follows the condition
\begin{equation}\label{condition}
F_x = 0, \,\, F_y = 0, \,\, and \,\, F_z = 0 
\end{equation}
and $($x$, $y$, $z$) \in [0, 1] \times [0, 1] \times [0, 1]$. By solving Eq. \ref{condition}, we can not only get the above eight fixed points, but also get that ($x$, $y$, $z$) = ($x^*$, $y^*$, $z^*$), which holds $x^* \in [0, 1]$, $y^* \in [0, 1]$, $z^* \in [0, 1]$, and
\begin{equation}\label{internal}
\begin{cases}
2z^*F+y^*R_1+y^*L_1-C_1 = 0\\
2z^*F+x^*R_1+x^*L_2+x^*\alpha(V_1-V_2)+\alpha V_2-C_2 = 0\\
x^*y^*R_2+P-C_3 = 0
\end{cases}
\end{equation}
, i.e., it is also the solution of Eq. \ref{condition}; therefore, it is one of the equilibrium points.

The result follows.
\end{proof}

The goal of our study is to find the ESS of the evolutionary game among these equilibrium points. Since ESS only appears in pure points, the mixed point ($x^*$, $y^*$, $z^*$) is first excluded. Then, we derive the following theorem, which gives the conditions for eight fixed points to become ESS.

\begin{center}
\begin{table*}[htbp]
\centering
\caption{The conditions in which eight fixed points can be ESS points}
\label{conditions}
\renewcommand\arraystretch{1.2}
\setlength{\tabcolsep}{15mm}
\begin{tabular}{cc}
\Xhline{1.0pt}
ESS & Conditions \\
\Xhline{1.0pt}
(0, 0, 0) & $C_1>0$, $\alpha V_2<C_2$, and $P<C_3$ \\
\Xhline{1.0pt}
(0, 0, 1) & $2F<C_1$, $2F+\alpha V_2<C_2$, and $C_3<P$ \\
\Xhline{1.0pt}
(0, 1, 0) & $R_1+L_1<C_1$, $C_2<\alpha V_2$, and $P<C_3$ \\
\Xhline{1.0pt}
(0, 1, 1) & $2F+R_1+L_1<C_1$, $C_2<\alpha V_2+2F$, and $C_3<P$; \\
\Xhline{1.0pt}
(1, 0, 0) & $C_1<0$, $R_1+L_2+\alpha V_1<C_2$, and $P<C_3$ \\
\Xhline{1.0pt}
(1, 0, 1) & $C_1<2F$, $2F+R_1+L_2+\alpha V_1<C_2$, and $C_3<P$ \\
\Xhline{1.0pt}
(1, 1, 0) & $C_1<R_1+L_1$, $C_2<R_1+L_2+\alpha V_1$, and $R_2+P<C_3$ \\
\Xhline{1.0pt}
(1, 1, 1) & $C_1<2F+R_1+L_1$ and $C_2<2F+R_1+L_2+\alpha V_1$, and $C_3<R_2+P$ \\
\Xhline{1.0pt}
\end{tabular}
\end{table*}
\end{center}

\begin{figure*}[htbp]
\begin{equation}  \label{Jacobi matrix}
\begin{normalsize}
\begin{aligned}
&J=\left( \begin{matrix}
	\frac{\partial F_x}{\partial x} & \frac{\partial F_x}{\partial y} & \frac{\partial F_x}{\partial z}\\
	\frac{\partial F_y}{\partial x} & \frac{\partial F_y}{\partial y} & \frac{\partial F_y}{\partial z}\\
	\frac{\partial F_z}{\partial x} & \frac{\partial F_z}{\partial y} & \frac{\partial F_z}{\partial z}\\
\end{matrix} \right) \\
&=\left( \begin{matrix}
	(1-2x)(2zF+yR_1+yL_1-C_1) & x(1-x)(R_1+L_1) & 2x(1-x)F\\
	\multirow{2}{*}{ $y(1-y)(R_1+L_2+\alpha (V_1-V_2))$ }& (1-2y)(2zF+xR_1+xL_2 & \multirow{2}{*}{ $2y(1-y)F$ }\\
         & +x\alpha (V_1-V_2)+\alpha V_2-C_2) & \\
	yz(1-z)R_2 & xz(1-z)R_2 & (1-2z)(xyR_2+P-C_3)\\
\end{matrix} \right)
\end{aligned}
\end{normalsize}
\end{equation}
\end{figure*}

\begin{thm}\label{ESS conditions}
(i) The equilibrium point (0, 0, 0) is the ESS, if $C_1>0$, $\alpha V_2<C_2$, and $P<C_3$;

(ii) The equilibrium point (0, 0, 1) is the ESS, if $2F<C_1$, $2F+\alpha V_2<C_2$, and $C_3<P$;

(iii) The equilibrium point (0, 1, 0) is the ESS, if $R_1+L_1<C_1$, $C_2<\alpha V_2$, and $P<C_3$;

(iv) The equilibrium point (0, 1, 1) is the ESS, if $2F+R_1+L_1<C_1$, $C_2<\alpha V_2+2F$, and $C_3<P$;

(v) The equilibrium point (1, 0, 0) is the ESS, if $C_1<0$, $R_1+L_2+\alpha V_1<C_2$, and $P<C_3$;

(vi) The equilibrium point (1, 0, 1) is the ESS, if $C_1<2F$, $2F+R_1+L_2+\alpha V_1<C_2$, and $C_3<P$;

(vii) The equilibrium point (1, 1, 0) is the ESS, if $C_1<R_1+L_1$, $C_2<R_1+L_2+\alpha V_1$, and $R_2+P<C_3$;

(viii) The equilibrium point (1, 1, 1) is the ESS, if $C_1<2F+R_1+L_1$ and $C_2<2F+R_1+L_2+\alpha V_1$.
\end{thm}

\begin{proof}
According to the method proposed by Friedman, the stability of the equilibrium points of a dynamic system described by a differential equation can be obtained from the local stability analysis of the Jacobi matrix. Based on the Eq. \ref{F_x}, the Jacobi matrix is expressed in Eq. \ref{Jacobi matrix}.

According to the evolutionary game, the equilibrium point is the ESS only if all the eigenvalues of the Jacobi matrix shown in Eq. \ref{Jacobi matrix} are negative. When the point is the internal equilibrium point ($x^*$, $y^*$, $z^*$), as stated in Eq. \ref{internal}, the trace of the Jacobian matrix is $tr(J)=(1-2x^*)(2z^*F+y^*R_1+y^*L_1-C_1)+(1-2y^*)(2z^*F+x^*R_1+x^*L_2+x^*\alpha (V_1-V_2)+\alpha V_2-C_2)+(1-2z^*)(x^*y^*R_2+P-C_3)=0$. And since the trace of the Jacobian matrix satisfies $tr(J)=\lambda_1+\lambda_2+\lambda_3=0$, where $\lambda_1$, $\lambda_2$, and $\lambda_3$ are the three eigenvalues of the Jacobi matrix when the equilibrium point is ($x^*$, $y^*$, $z^*$), which cannot make all three eigenvalues negative, i.e., it is impossible for the internal equilibrium point ($x^*$, $y^*$, $z^*$) to be ESS. Subsequently, we prove the conditions for the other eight fixed equilibrium points to become ESS.

(i) For the equilibrium point (0, 0, 0), we get that the three eigenvalues of the Jacobi matrix at this time are $\lambda_1=-C_1$, $\lambda_2=\alpha V_2-C_2$, and $\lambda_3=P-C_3$. If we expect (0, 0, 0) to be ESS, then the conditions $\lambda_1<0$, $\lambda_2<0$, and $\lambda_3<0$ should be followed, i.e., we yield $C_1>0$, $\alpha V_2<C_2$, and $P<C_3$;

(ii) For the equilibrium point (0, 0, 1), we obtain that the three eigenvalues of the Jacobi matrix at this time are $\lambda_1=2F-C_1$, $\lambda_2=2F+\alpha V_2-C_2$, and $\lambda_3=C_3-P$. If we expect (0, 0, 1) to be ESS, then the conditions $\lambda_1<0$, $\lambda_2<0$, and $\lambda_3<0$ should be satisfied, i.e., we have $2F<C_1$, $2F+\alpha V_2<C_2$, and $C_3<P$;

Analogously, we can derive conditions for the other 6 equilibrium points to become ESS.

The results follow.
\end{proof}

Therefore, the conditions in which the eight fixed points are ESS points are displayed in Tab. \ref{conditions}. It is worth noting that in contrast to symmetric games, such as the Chicken game and the Snowdrift game, the internal equilibrium point of the asymmetric game presented in the model of this paper is difficult to become stable (absorbing), and it is not an ESS, which is in line with previous studies \cite{habib2022evolutionary, tanimoto2021sociophysics}.

We combine the actual situation of the three-party game in reality. The cost of providing high quality data to the data provider should be greater than 0, i.e., $C_1>0$. Hence, the condition that the equilibrium point (1, 0, 0) becomes ESS cannot be followed. Besides, the data regulator regulates the quality of open data and needs to hold that the cost paid to the data regulator to regulate it should be greater than the cost of the data regulator so that it can effectively promote sharing, i.e., it needs to satisfy $P>C_3$. Therefore, the conditions for equilibrium points (0, 0, 0), (0, 1, 0), and (1, 1, 0) to become ESS are not held. Then, the proposed model has a total of four ESSs when their corresponding conditions are fulfilled.

This ends the modeling section. In brief, we first described the evolutionary game model for the stakeholders in open data, then illustrated the calculation of players' payoffs after making some necessary assumptions and notations, and finally discussed the model to obtain the conditions that an equilibrium point becomes ESS based on the replicator dynamic equations.

\section{Numerical simulation}  \label{Simulation}

In this section, with the purpose of confirming the previous theory, we use numerical simulations to characterize the evolutionary path and analyze the sensitivity of the models. In order to solve the replicator dynamic equations of the model presented in Eq. \ref{F_x} and intuitively observe the dynamic evolution of the strategies of the three players, we employ the functions $integrate.odeint$ of the package $scipy$ and $pyplot$ of the package $matplotlib$ in Python system simulation tool, respectively.

\subsection{The Evolutionary Trajectories of Three Players}

\begin{center}
\begin{figure*}[htp]
\begin{minipage}{0.5\linewidth}
\vspace{3pt}
\hspace{5pt}
\centerline{\includegraphics[scale = 0.5]{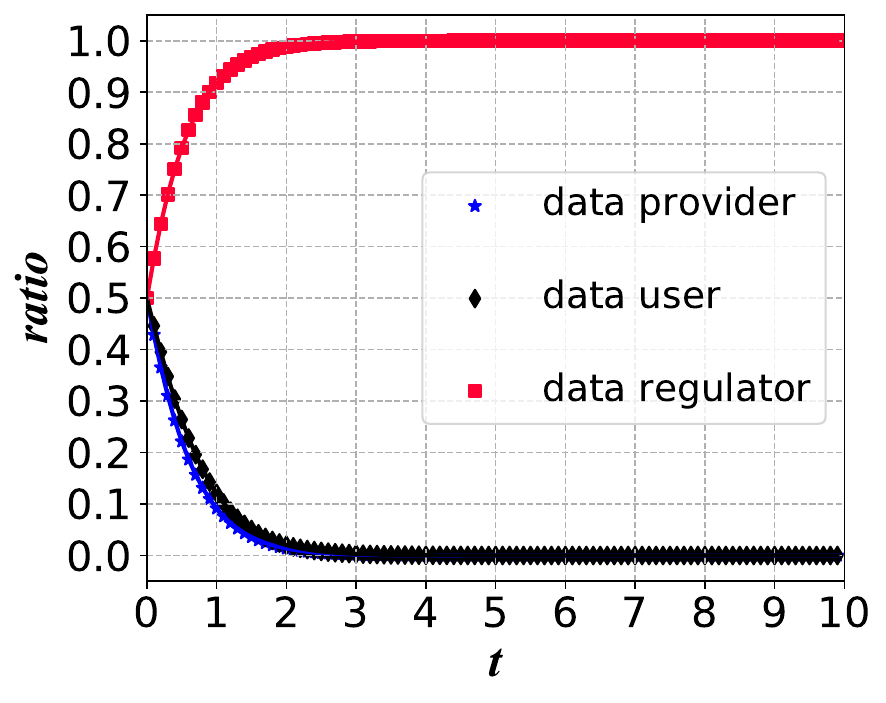}}
\centerline{(a) ESS(0, 0, 1)}
\end{minipage}
\begin{minipage}{0.55\linewidth}
\label{ESS(0,0,1)}
\vspace{3pt}
\hspace{5pt}
\centerline{\includegraphics[scale = 0.5]{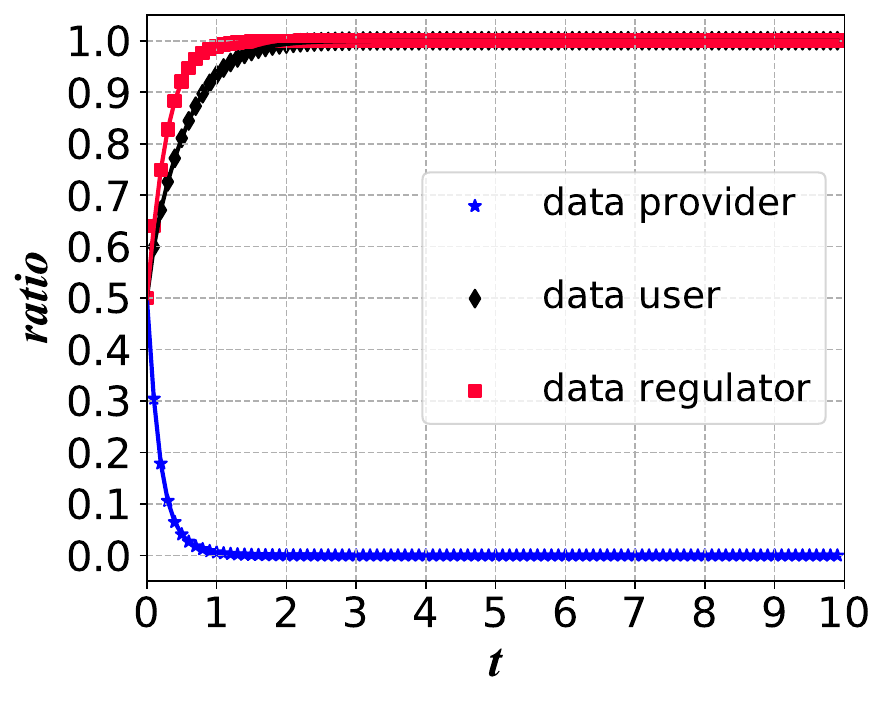}}
\label{ESS(0,1,1)}
\centerline{(b) ESS(0, 1, 1)}
\end{minipage}
\begin{minipage}{0.5\linewidth}
\label{ESS(1,0,1)}
\vspace{3pt}
\centerline{\includegraphics[scale = 0.5]{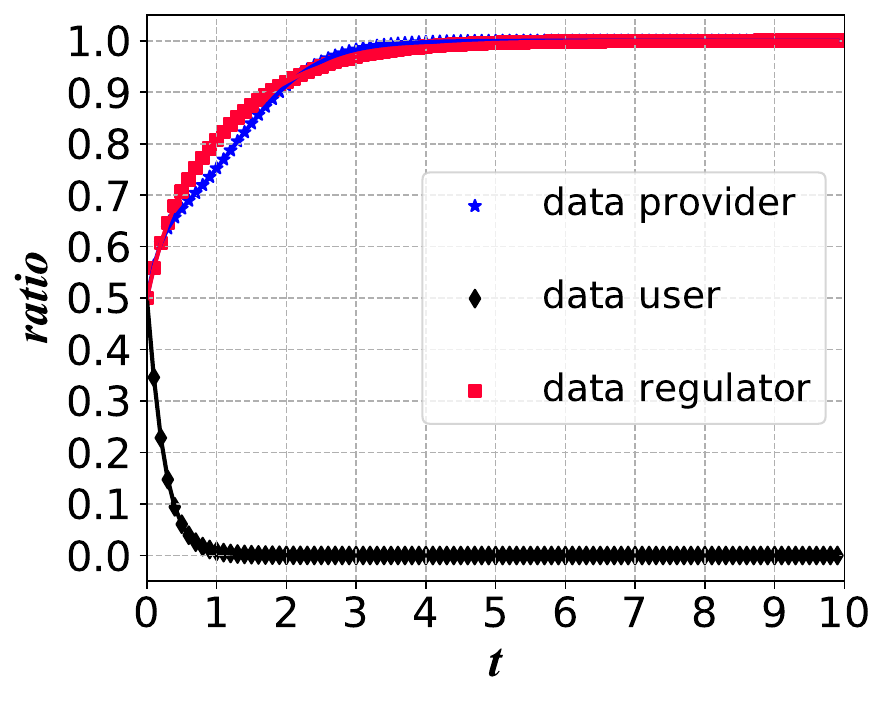}}
\centerline{(c) ESS(1, 0, 1)}
\end{minipage}
\begin{minipage}{0.55\linewidth}
\label{ESS(1,1,1)}
\vspace{3pt}
\centerline{\includegraphics[scale = 0.5]{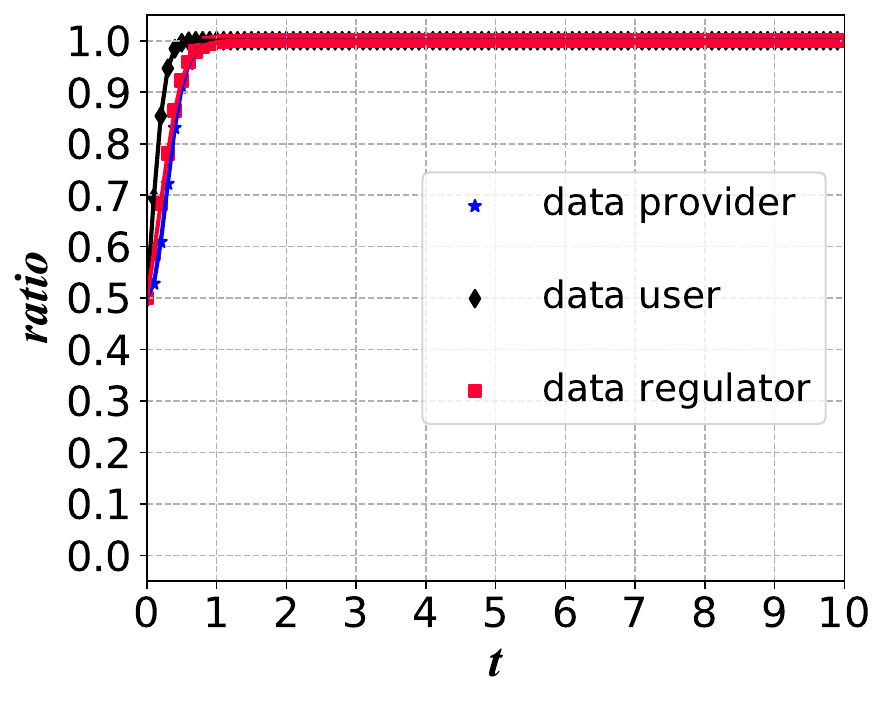}}
\centerline{(d) ESS(1, 1, 1)}
\end{minipage}
\caption{\textbf{The Evolutionary Trajectories of Three Players.} The $x$-axis and the $y$-axis are set as the time and the cooperation ratio separately. (a) The parameters are set as $C_1 = 8, C_2 = 10, C_3 = 6, \alpha = 0.7, V_1 = 5, V_2 = 3, R_1 = 2, L_1 = 2, L_2 = 2, P = 8, R_2 = 5$, and $F = 3$ to follow the conditions for (0, 0, 1) to become ESS. (b) The parameters are set as $C_1 = 15, C_2 = 10, C_3 = 6, \alpha = 0.8, V_1 = 10, V_2 = 8, R_1 = 4, L_1 = 2, L_2 = 5, P = 10, R_2 = 8$, and $F = 3$ to satisfy the condition for (0, 1, 1) to become ESS. (c) The parameters are set as $C_1 = 4, C_2 = 13, C_3 = 6, \alpha = 0.2, V_1 = 8, V_2 = 5, R_1 = 3, L_1 = 5, L_2 = 1, P = 7, R_2 = 6$, and $F = 3$ to obey the condition for (1, 0, 1) to become ESS. (d) The parameters are set as $C_1 = 8, C_2 = 6, C_3 = 7, \alpha = 0.7, V_1 = 10, V_2 = 8, R_1 = 4, L_1 = 4, L_2 = 3, P = 9, R_2 = 5$, and $F = 4$ to meet the two different conditions for (1, 1, 1) to become ESS. It can be seen that the results of the numerical simulations are in good agreement with our previously derived theory.}
\label{ESS verify}
\end{figure*}
\end{center}

Without loss of generality, we initially make all strategies evenly mixed among the players' choices, i.e., the initial fractions of $x$, $y$, and $z$ equal to 0.5. Here, we have symbolically assigned the parameters in order to follow the ESS conditions discussed in Sec. \ref{Model discussion}.

\subsubsection{The Evolutionary Trajectories of Players at ESS(0, 0, 1)}

When $2F<C_1$, $2F+\alpha V_2<C_2$, and $C_3<P$, it follows from Thm. \ref{ESS conditions}(ii) that (0, 0, 1) is ESS, i.e., when the above three inequalities are true, data providers adopt the strategy of opening low quality data, data users choose not to acquire data, and data regulators select to regulate open data. By setting the $x$-axis as time and the $y$-axis as the ratio of cooperation, we present the evolutionary trajectories of players in Fig. 2(a). Besides, other parameters are set as $C_1 = 8, C_2 = 10, C_3 = 6, \alpha = 0.7, V_1 = 5, V_2 = 3, R_1 = 2, L_1 = 2, L_2 = 2, P = 8, R_2 = 5$, and $F = 3$ to meet the above conditions. From Fig. 2(a), we can demonstrate a gradual decrease to 0 in the proportion of data providers selecting to open high quality data and data users performing to acquire data, and a gradual increase to 1 in the proportion of data regulators selecting to regulate open data. In addition, all three curves reach a steady state after $t = 2$.

\subsubsection{The Evolutionary Trajectories of Players at ESS(0, 1, 1)}

When $2F+R_1+L_1<C_1$, $C_2<\alpha V_2+2F$, and $C_3<P$, it follows from Thm. \ref{ESS conditions}(iv) that (0, 1, 1) is ESS, i.e., when the above three inequalities follow, data providers choose not to open high quality data, data users adopt the strategy of acquiring data, and data regulators select to regulate open data. The parameters are set as $C_1 = 15, C_2 = 10, C_3 = 6, \alpha = 0.8, V_1 = 10, V_2 = 8, R_1 = 4, L_1 = 2, L_2 = 5, P = 10, R_2 = 8$, and $F = 3$ to meet the above conditions, and the evolutionary trajectories of players are shown in Fig. 2(b). Hence, the ratio of data providers choosing to open high quality data consistently decreases to 0, whereas the ratio of data users selecting to acquire data and data regulators performing to regulate open data constantly increases to 1 as time progresses. Additionally, the strategies of data providers reach stability in a shorter period of time than that of the data users and regulators.

\subsubsection{The Evolutionary Trajectories of Players at ESS(1, 0, 1)}

For $C_1<2F$, $2F+R_1+L_2+\alpha V_1<C_2$, and $C_3<P$, it follows from Thm. \ref{ESS conditions}(vi) that (1, 0, 1) is ESS, i.e., when the above three inequalities are true, data providers perform to open high quality data, data users choose not to acquire data, and data regulators select to regulate open data. The evolutionary trajectories of players are demonstrated in Fig. 2(c), where the parameters are set as $C_1 = 4, C_2 = 13, C_3 = 6, \alpha = 0.2, V_1 = 8, V_2 = 5, R_1 = 3, L_1 = 5, L_2 = 1, P = 7, R_2 = 6$, and $F = 3$ to follow the above conditions. It can be observed that the ratio of data providers selecting to open high quality data and data regulators performing to regulate open data keeps rising and finally reaches 1, while the ratio of data users choosing to acquire data keeps decreasing and finally reaches 0. The data user here reaches a steady state faster than the data provider and regulator at this situation.

\subsubsection{The Evolutionary Trajectories of Players at ESS(1, 1, 1)}

\begin{figure*}[htbp]
\centering
\subfigure[$C_2$ vs $\alpha$]{
\includegraphics[scale=0.25]{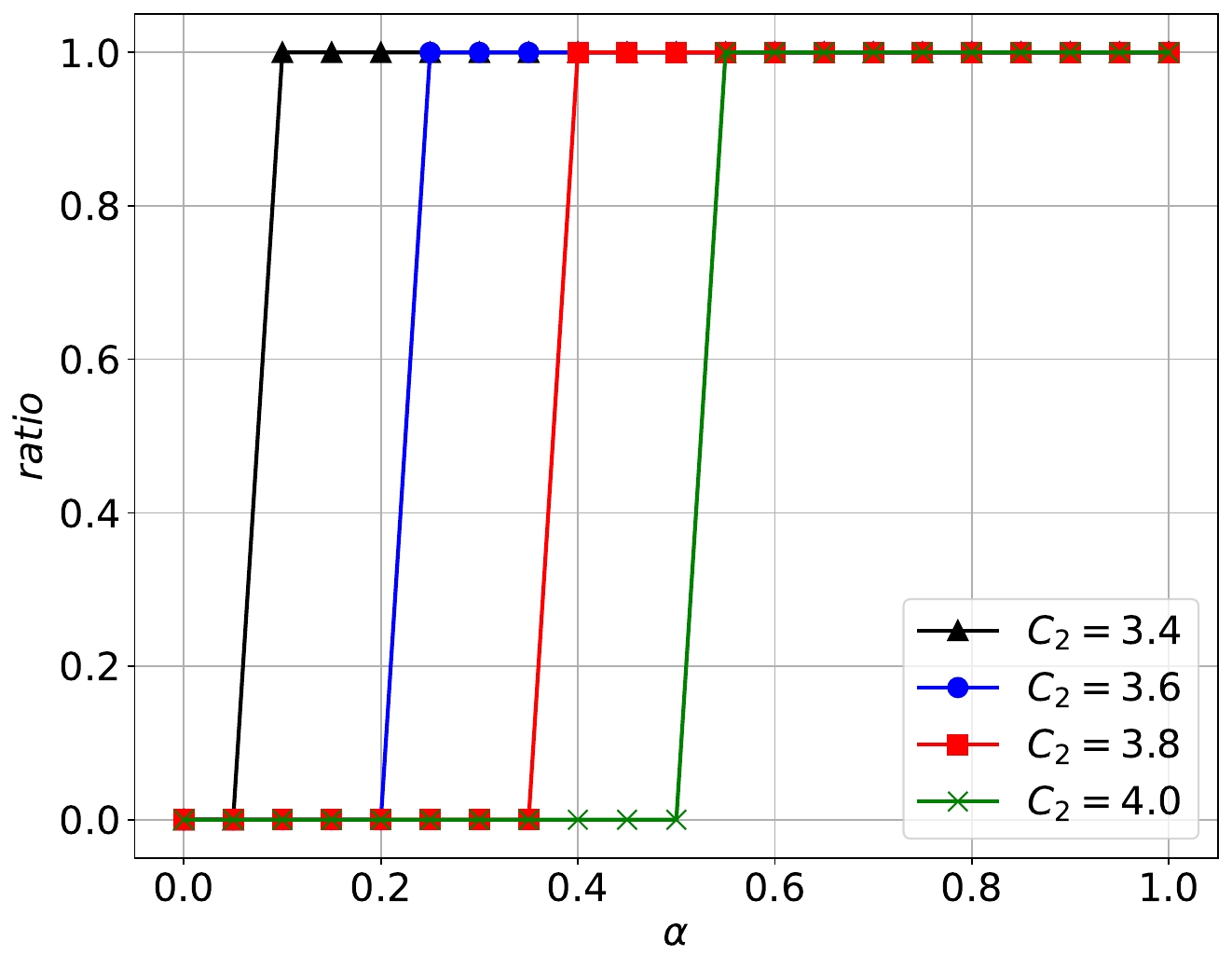}
\label{c2alpha}}
\subfigure[$V_1$ vs $\alpha$]{
\includegraphics[scale=0.25]{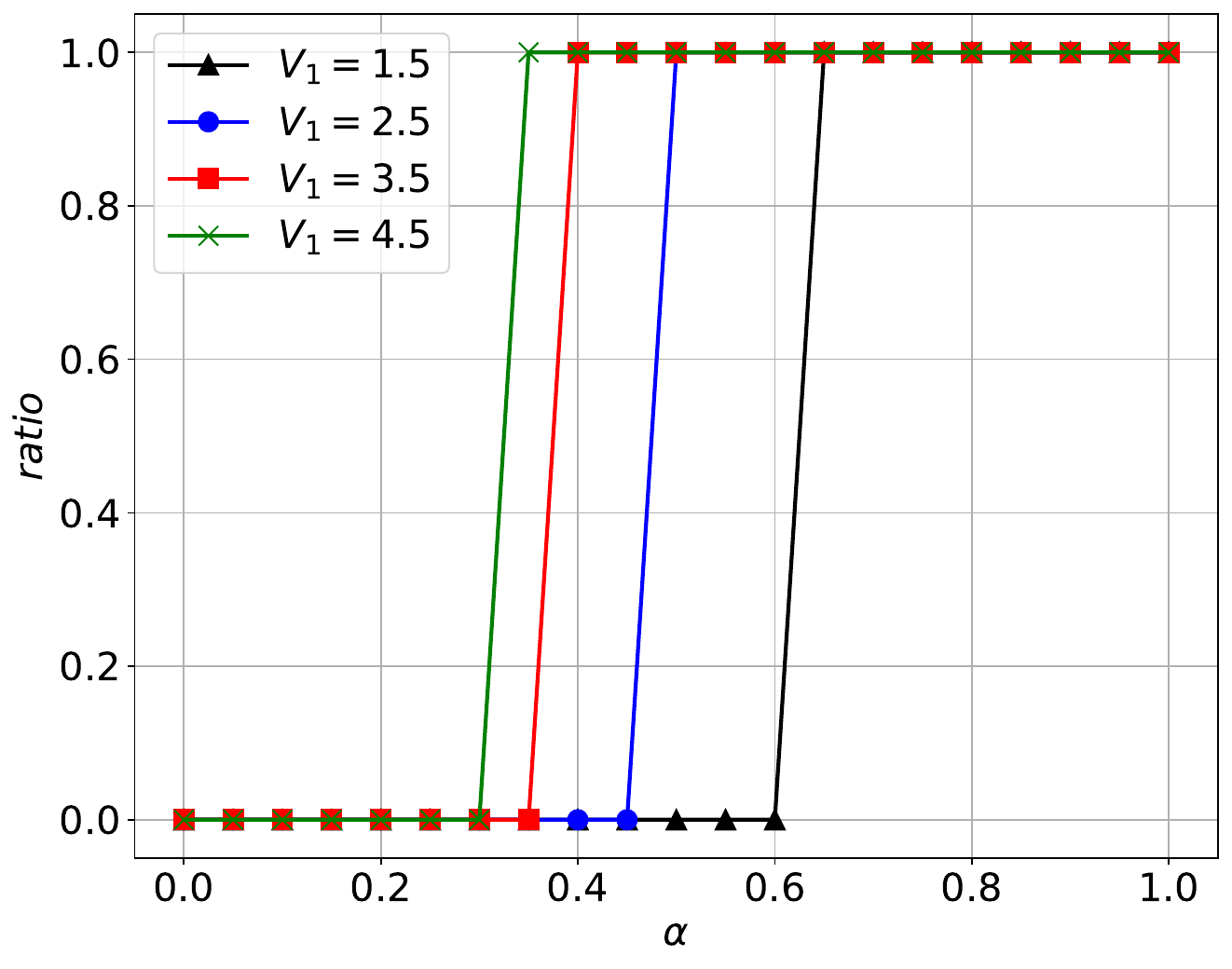}
\label{v1alpha}}
\subfigure[$F$ vs $\alpha$]{
\includegraphics[scale=0.25]{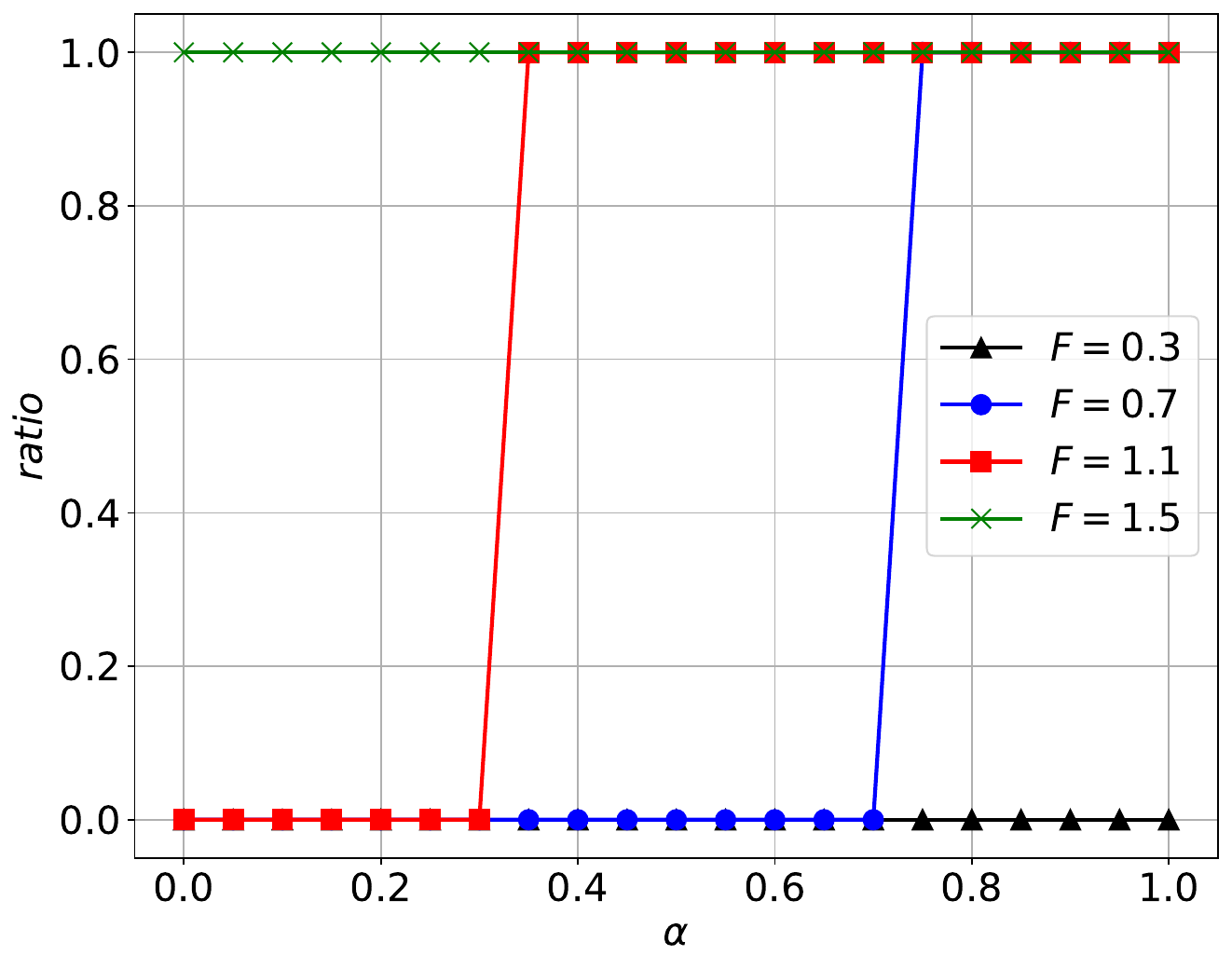}
\label{falpha}}
\caption{\textbf{Effects of $\alpha$, $C_2$, $V_1$, and $F$ on the fraction of the data users choosing to acquire data.} Here, we fix $C_1=6$, $C_3=5$, $V_2=1$, $R_1=2$, $L_1=5$, $L_2=7$, $P=3$, $R_2=4$. Regarding the concerned $C_2$, $V_1$, and $F$, we have (a) $C_2= [3.4, 3.6, 3.8, 4.0]$, $V_1=2$, $F=1$. (b) $C_2=4$, $V_1=[1.5, 2.5, 3.5, 4.5]$, and $F=1$. (c) $C_2=4$, $V_1=3$, and $F=[0.3, 0.7, 1.1, 1.5]$. The $x$ and $y$ axes are set as the data mining capability $\alpha$ and the fraction of the data users adopting to acquire data, respectively. It can be observed that a larger value of data mining capability $\alpha$ results in a greater cooperation level of the data user.}
\label{X_alpha}
\end{figure*}

\begin{figure*}[htbp]
\centering
\subfigure[$\alpha = 0.30$]{
\includegraphics[scale=0.24]{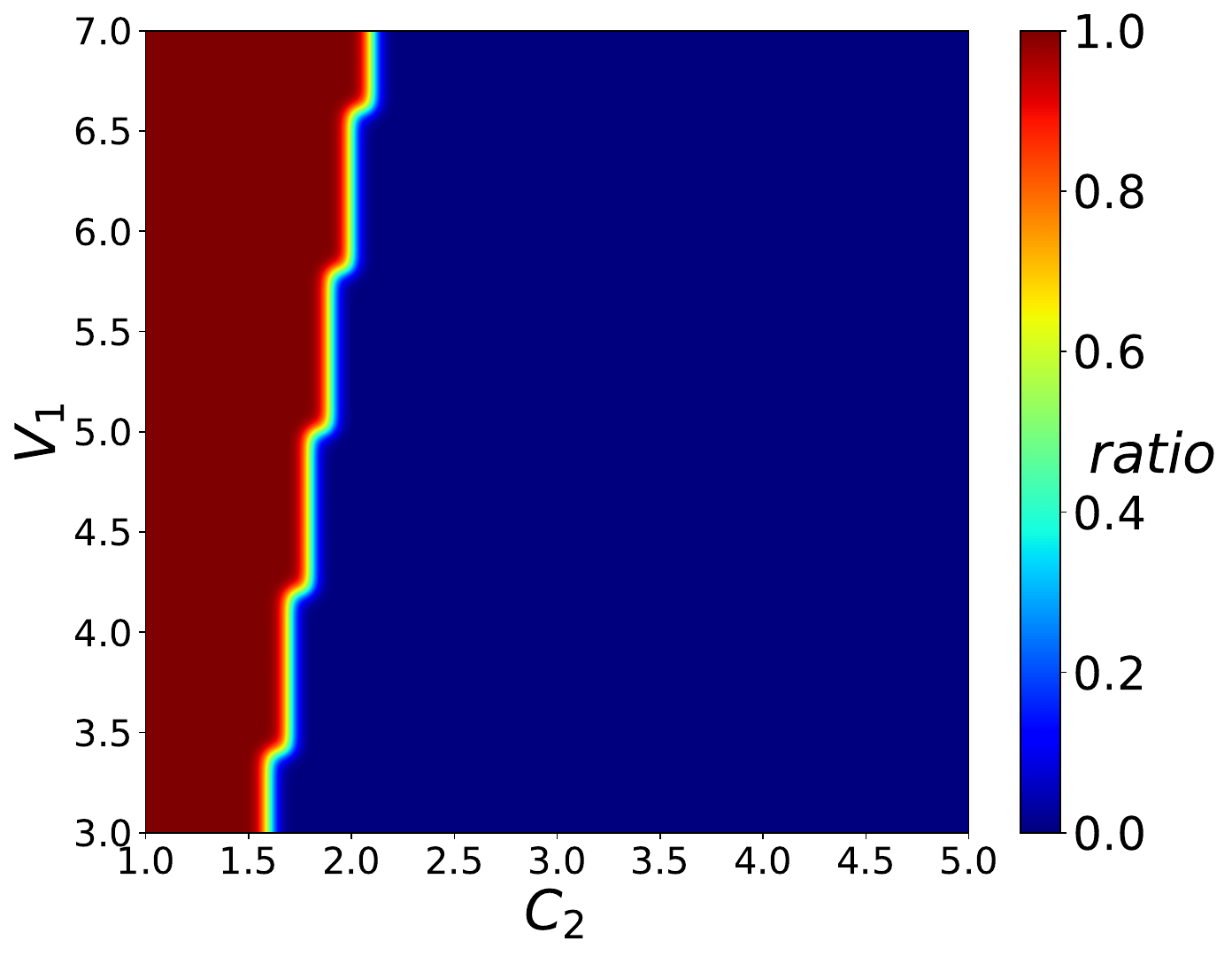}
\label{heatmaps1}}
\subfigure[$\alpha = 0.60$]{
\includegraphics[scale=0.24]{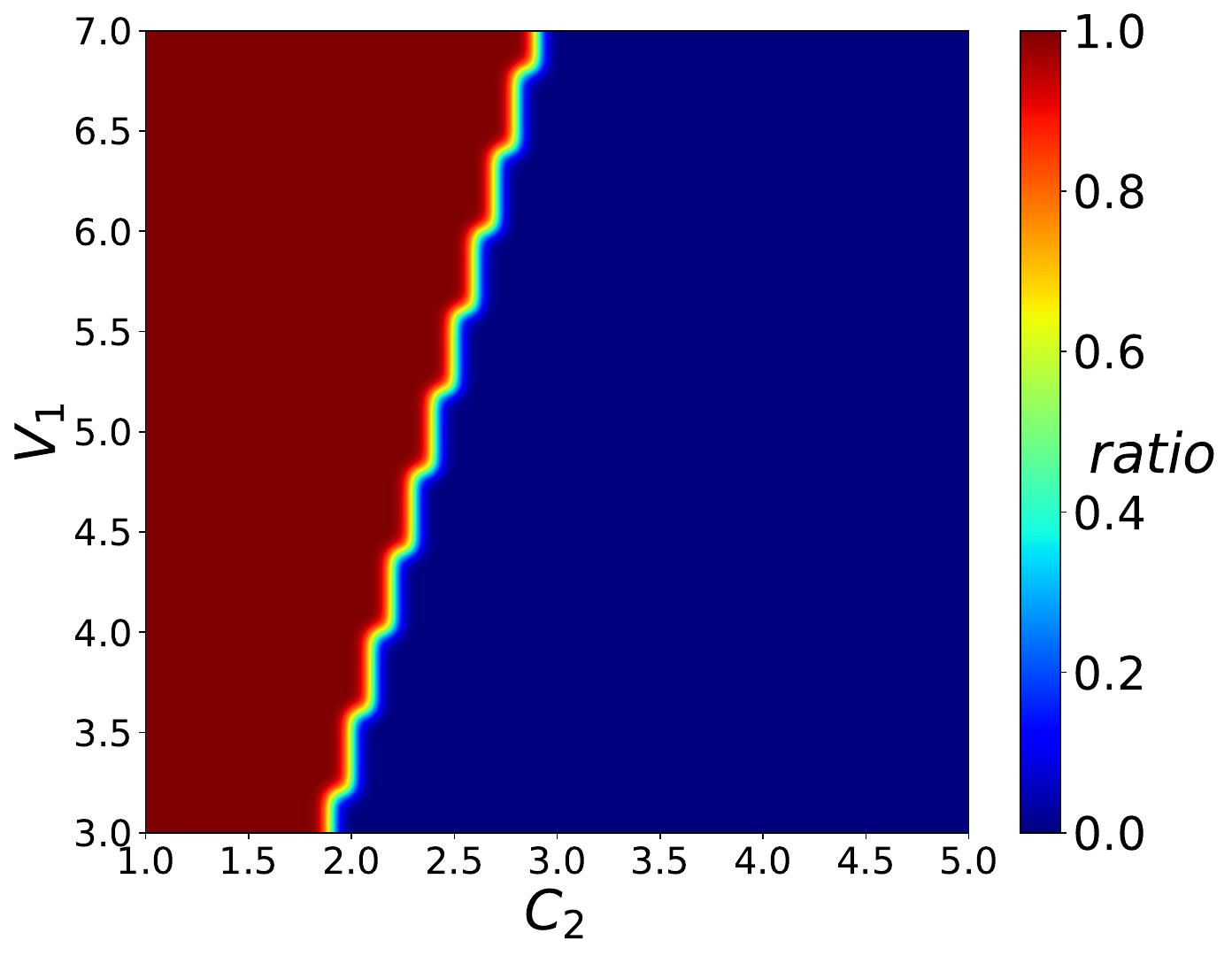}
\label{heatmaps2}}
\subfigure[$\alpha = 0.90$]{
\includegraphics[scale=0.24]{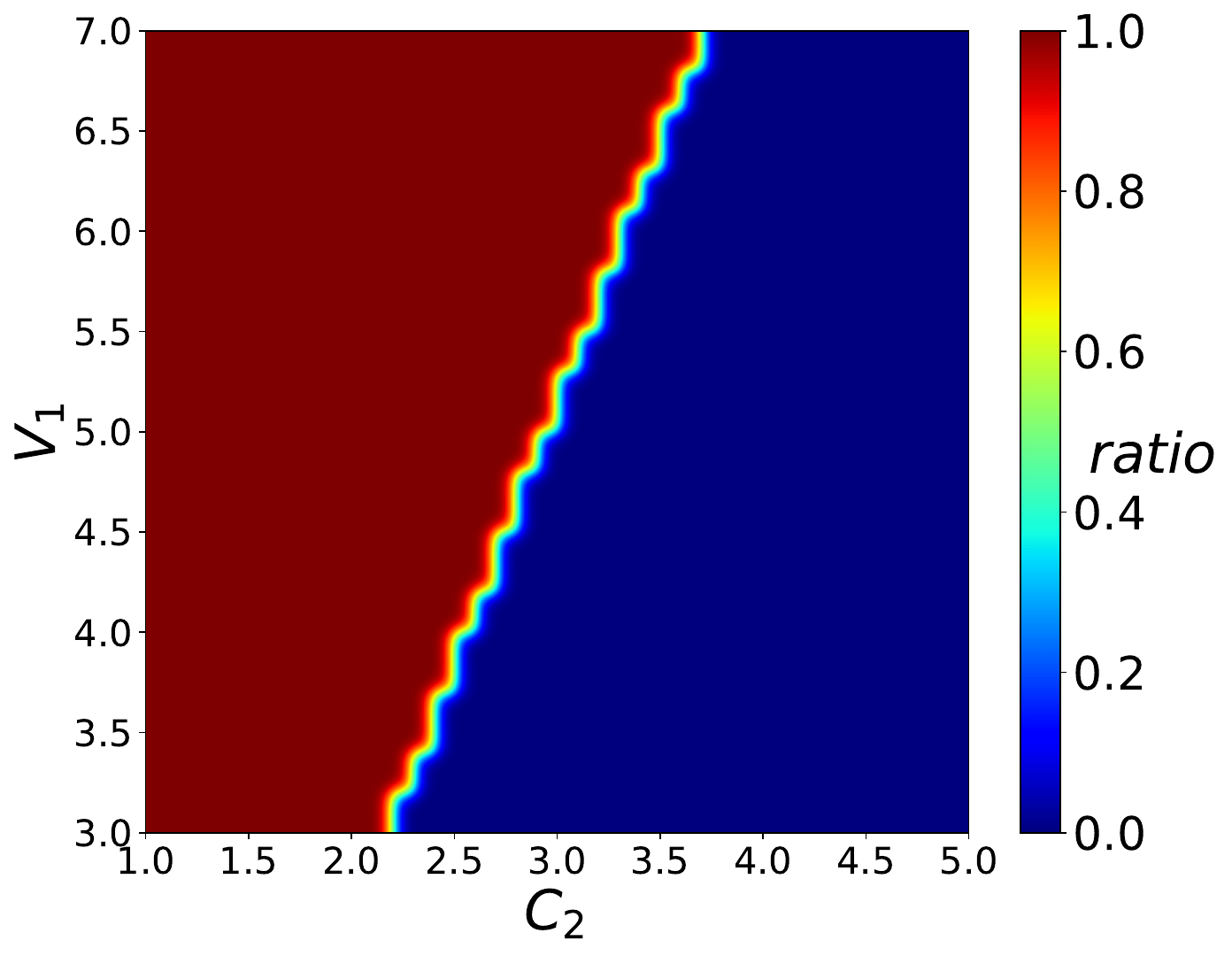}
\label{heatmaps3}}
\caption{\textbf{Heat maps of the fraction of the data user performing to acquire data considering $V_1$ and $C_2$.} Here, we set $V_1$ and $C_2$ in $[1,5]$ and $[3,7]$ respectively, and fix $C_1=6$, $C_3=5$, $V_2=1$, $R_1=2$, $L_1=5$, $L_2=7$, $P=3$, $R_2=4$, and $F=1$. Besides, we have (a) $\alpha=0.30$, (b) $\alpha=0.60$, and (c) $\alpha=0.90$. The color bar is in the range $[0, 1]$ as shown in each subgraph. For the larger data mining capability $\alpha$, the ratio of the data user choosing to acquire data can be enhanced by decreasing the cost of the data users acquiring data $C_2$ or increasing the value of high quality data $V_1$.}
\label{heatmaps}
\end{figure*}

When $C_1<2F+R_1+L_1$, $C_2<2F+R_1+L_2+\alpha V_1$, and $C_3<R_2+P$, it follows from Thm. \ref{ESS conditions}(viii) that (1, 1, 1) is ESS since the condition $C_3<R_2+P$ is always satisfied, i.e., when the above two inequalities follow, data providers perform to open high quality data, data users select to acquire data, and data regulators choose to regulate open data. The parameters are set as $C_1 = 8, C_2 = 6, C_3 = 7, \alpha = 0.7, V_1 = 10, V_2 = 8, R_1 = 4, L_1 = 4, L_2 = 3, P = 9, R_2 = 5$, and $F = 4$ to meet the above conditions, and the evolutionary trajectories of players are shown in Fig. 2(d). We can see that the ratio of data providers performing to open high quality data, data users selecting to acquire data, and data regulators choosing to regulate open data consistently increase to 1. Moreover, Fig. 2(d) illustrates that the evolution of the data user reaches the steady state faster than the data provider and regulator, while the data provider and regulator reach the steady state at almost the same time and the steady state time of all three curves is less than 1.

From the evolutionary trajectories of players at different evolutionary stable strategies, we find that the results of numerical simulations are consistent with our previous theory.

\subsection{Effects of the Data Mining Capability on Data Users' Strategies}
\label{Effects of the data mining capability on data users' strategy}

\begin{figure*}[htbp]
\centering
\subfigure[$\alpha = 0.30$, $C_2$]{
\includegraphics[scale=0.52]{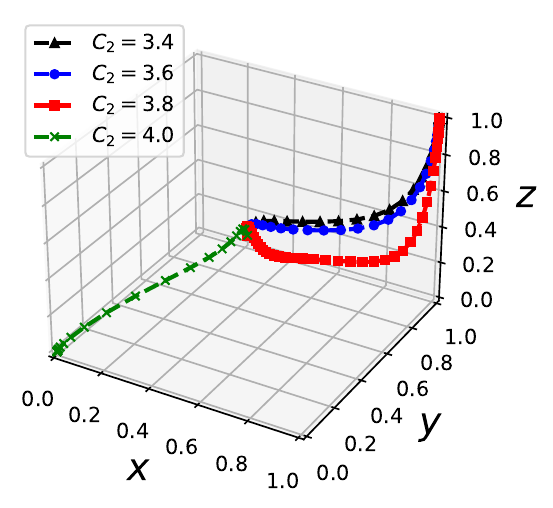}
\label{tracec2alpha1}}
\subfigure[$\alpha = 0.60$, $C_2$]{
\includegraphics[scale=0.52]{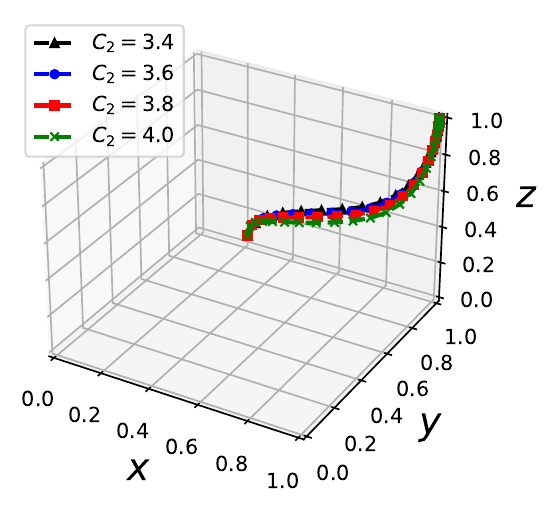}
\label{tracec2alpha2}}
\subfigure[$\alpha = 0.90$, $C_2$]{
\includegraphics[scale=0.52]{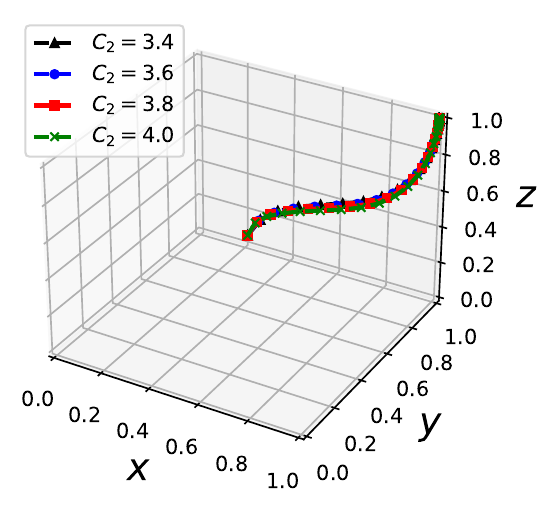}
\label{tracec2alpha3}}

\subfigure[$\alpha = 0.30$, $V_1$]{
\includegraphics[scale=0.52]{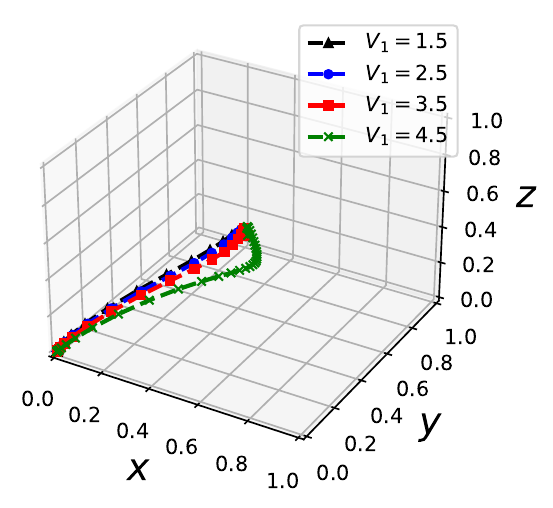}
\label{tracev1alpha1}}
\subfigure[$\alpha = 0.60$, $V_1$]{
\includegraphics[scale=0.52]{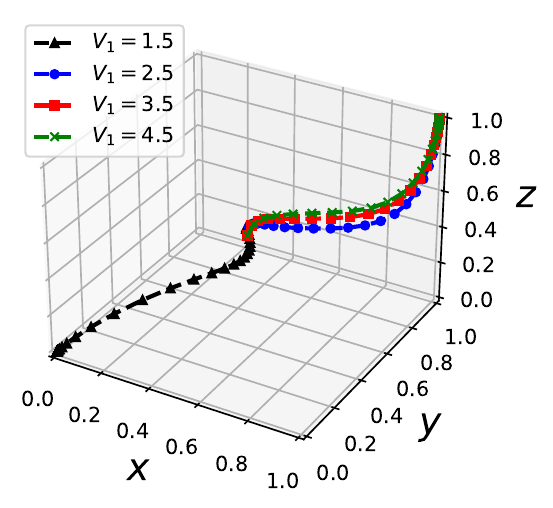}
\label{tracev1alpha2}}
\subfigure[$\alpha = 0.90$, $V_1$]{
\includegraphics[scale=0.52]{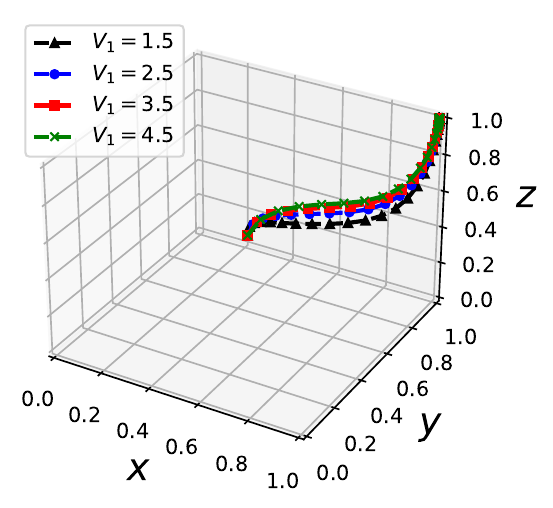}
\label{tracev1alpha3}}

\subfigure[$\alpha = 0.30$, $F$]{
\includegraphics[scale=0.52]{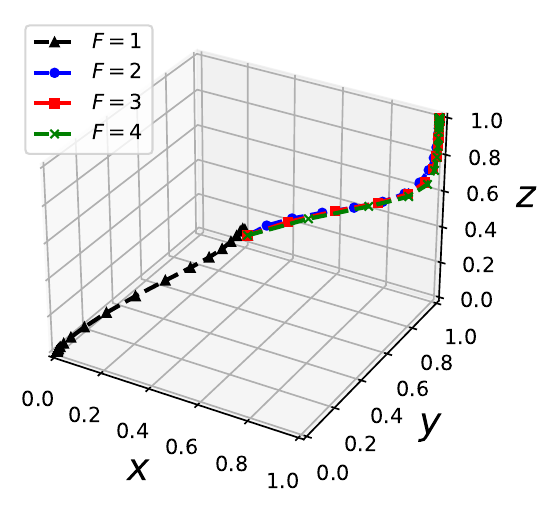}
\label{tracefalpha1}}
\subfigure[$\alpha = 0.60$, $F$]{
\includegraphics[scale=0.52]{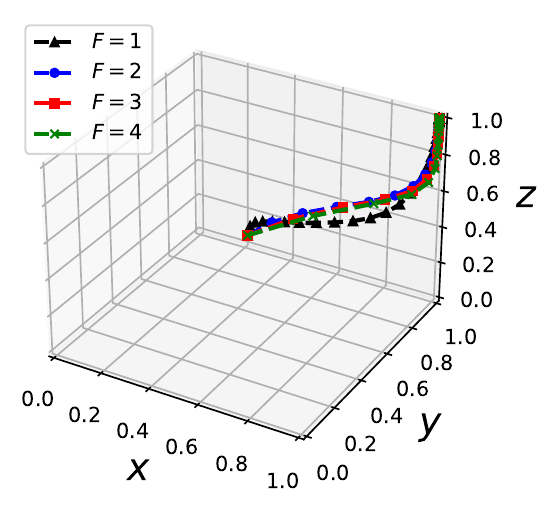}
\label{tracefalpha2}}
\subfigure[$\alpha = 0.90$, $F$]{
\includegraphics[scale=0.52]{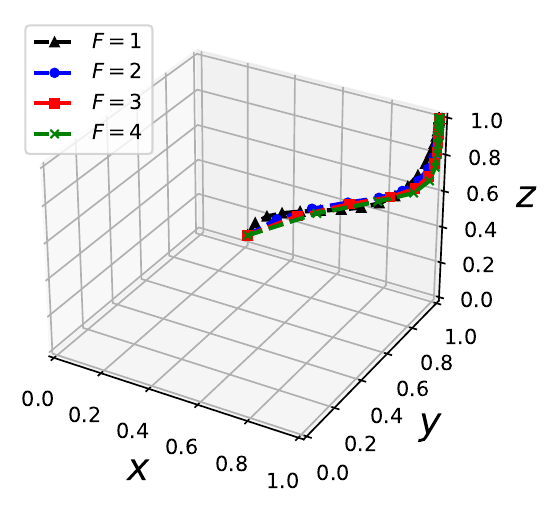}
\label{tracefalpha3}}
\caption{\textbf{3D traces of $x$, $y$, and $z$ in the evolution processes.} Here, the fixed parameters are $C_1=6$, $C_3=5$, $V_2=1$, $R_1=2$, $L_1=5$, $L_2=7$, $P=3$, and $R_2=4$. $\alpha\in[0.30, 0.60, 0.90]$ is set for the cross simulation. In (a), (b), and (c), we set $C_2=[3.4, 3.6, 3.8, 4.0]$, $V_1=3$, and $F=1$. In (d), (e), and (f), we set $V_1=[1.5, 2.5, 3.5, 4.5]$, $C_2=4$, and $F=1$. In (g), (h), and (i), we set $F=[1, 2, 3, 4]$, $C_2=4$, and $V_1=3$. Each curve stands for a changing trajectory of $x$, $y$, and $z$, of which the initial values are all $0.5$. The 3D traces exhibit that an increase in the data mining capability $\alpha$ will provide more relaxed conditions for the whole group to adopt cooperative strategies.}
\label{traces}
\end{figure*}

In this subsection, we explore how the data mining capability affects data users' strategies along with three other factors, including $C_2$, $F$, and $V_1$. The reason for choosing these three factors is their significant impact on the group strategy with the users' data mining capability $\alpha$, which can indicate meaningful and noticeable results.

As shown in Fig. \ref{X_alpha}, we represent the fraction of data users choosing to acquire data against the data mining capability. Each data point is obtained from a stable state of the evolution system. From Fig. \ref{c2alpha}, it is obvious that the increase of $C_2$ brings a higher capability threshold for data users to choose to acquire the data. The same pure strategy state can be reached if different $\alpha$s are given, but the cost of the data users to acquire data needs to be minimized to provide the population with a large space for positive interaction. From Fig. \ref{v1alpha}, we find that the increase of $V_1$ is beneficial for data users to acquire data. However, this positive effect is less significant if the original $V_1$ is already large enough. In Fig. \ref{falpha}, the fraction of data users who choose to acquire data is sensitive to the parameter $F$. We find that minor changes of $F$ induce a sudden change in the data users' strategies, and this change influences the group strategy configuration conspicuously. 

Subsequently, to further explore the influences of $V_1$ and $C_2$ on the data users' strategies, we present the heat maps (also called phase diagrams) shown in Fig. \ref{heatmaps}. We set $V_1$ in the range $[3, 7]$ and $C_2$ in the range $[1, 5]$ to investigate the fraction of data users to acquire data and the phase transition for different $\alpha s\in[0.30, 0.60, 0.90]$. According to the results in Figs. \ref{heatmaps1}, \ref{heatmaps2}, and \ref{heatmaps3}, the fraction of data users performing to acquire data undergoes a discontinuous phase transition from 1 to 0 with the decrease of $V_1$ and the increase of $C_2$. For the data mining capability $\alpha=0.60$ and $0.90$, a higher $C_2$ and smaller $V_1$ than Fig. \ref{heatmaps1} ($\alpha=0.30$) are allowed to bring more data users to select the strategy of acquiring data. Combined with our previous theoretical analysis, we can conclude that the increase in the data mining capability $\alpha$ provides more lenient conditions for the data users to acquire data effectively and then drive the interaction in the group. 

\subsection{The Impacts of Parameters on Evolutionary Results}

By showing the changing traces of $x$, $y$, and $z$ in the three-dimensional space, we explore the strategy form of the system in the evolution process. Without loss of generality, the initial fractions of the three strategies are all 0.5.

For the same reason as above, we here choose $\alpha$, $C_2$, $V_1$, and $F$ as the arguments and show the results in Fig. \ref{traces}. For a better presentation, we assume that the three-party system reaches the cooperation state if all data providers choose to open high quality data, all data users perform to acquire data, and all data regulators select to regulate the open data. As the increase of the data users' cost $C_2$ to acquire data in Fig. \ref{tracec2alpha1}, the fractions of data providers choosing to open high quality data ($x$), the data users performing to acquire data ($y$), and the data regulators selecting to regulate the open data ($z$) all decrease to $0$. However, along with the increase of the data mining capability $\alpha$ in Figs. \ref{tracec2alpha2} and \ref{tracec2alpha3}, $x$, $y$, and $z$ all increase to $1$ even for a large $C_2$, which enhances the cooperation among the three-party. In Fig. \ref{tracev1alpha1}, the value of high quality data $V_1$ we set is not enough to motivate the data users to acquire data and encourage the data providers and regulators to respond positively. With the increase of the data mining capability $\alpha$, although the cooperation state is not reached if $V_1=1.5$ and $\alpha=0.60$, it is reached given any other situation in Figs. \ref{tracev1alpha2} and \ref{tracev1alpha3}. In Figs. \ref{tracefalpha1}, \ref{tracefalpha2}, and \ref{tracefalpha3}, we explore the influence of $\alpha$ and $F$. As shown in Fig. \ref{tracefalpha1}, the cooperation state may not be reached if $F$ is low. Additionally, with the increase of $\alpha$ and $F$, the cooperation state of the three-party game can be reached in a large parameter space.

\section{Conclusion and outlook} \label{Conclusion}

In today's competitive environment, a good open data relationship, where the data providers adopt the strategy of opening high quality data, data users select to acquire data with a positive attitude, and data regulators perform to regulate the open data, is very crucial in the development of the digital economy. In this study, we conduct an evolutionary game analysis on the provision, acquisition, and regulation of open data from the perspective of game theory for the groups of data providers, users, and regulators, and reach relevant conclusions and insights, which have certain guiding significance for optimizing the quality of open data, promoting wider public participation and the development of the digital economy. Different from previous studies, we introduce the data regulator and data mining capability of data user, which are vital and reasonable for realistic scenarios, and propose a novel evolutionary game model to further evaluate the relationship between multi-stakeholders (data providers, users, and regulators). More specifically, we first analyze the payoff matrix among the three-party game under different strategies based on reasonable assumptions, then develop a replicator dynamic system, and solve its ESS problem. We emphasize that only the 8 pure strategy combinations in our model are ESS, while the internal equilibrium point is not ESS, which is proven from the point of mathematical standpoint. Additionally, we conduct simulation experiments to further validate the theoretical results of the model and reveal that they are in good accordance with each other. 

Based on our theoretical analysis, we determine that cooperative strategies are adopted by all three groups under the following conditions: 

(1) The cost of the data providers to supply high quality data is lower than the combined sum of two times the rewards or penalties for both data providers and users under regulation, the synergistic benefits associated with data users accessing high quality data, and the losses incurred by data providers when low quality data is obtained by users ($C_1<2F+R_1+L_1$); 

(2) The cost of data users to acquire data is less than the combined sum of twice the rewards or penalties for both data providers and users under regulation, the synergistic benefits associated with data users accessing high quality data, the losses faced by data users when high quality data is accessed by competitors, and the data mining capability of data users multiplied by the value of the high quality data ($C_2<2F+R_1+L_2+\alpha V_1$); 

(3) The cost for data regulators to regulate open data is lower than the reward for data regulators executing regulation given that data providers provide high quality data and data users obtain data, in addition to the income generated for data regulators opting to regulate open data ($C_3<R_2+P$). 

Furthermore, based on the outcomes of our conducted numerical simulations, we can provide the following recommendations to steer the development of open data in a positive direction: 

(1) Data providers, can take some ways to improve the value of open data, such as starting with refining the original data to ensure accuracy and reliability, establishing a standardized system of norms and guidelines to maintain consistency and compatibility, and minimizing data redundancy and eliminating noise to enhance the quality of the provided data; 

(2) Data users, can work on strengthening their data mining capabilities and reducing data acquisition costs by integrating rapidly advancing technologies such as cloud computing and deep learning with their existing methods. This fusion of technologies can lead to more efficient data utilization; 

(3) For data regulators, it is recommended that data regulators consider augmenting the rewards or penalties imposed on both data providers and users under regulatory frameworks, which can effectively incentivize data providers and users to embrace a cooperative strategy. This multifaceted approach allows individuals from all three groups to collectively improve their practices based on their own perspectives, thus contributing to the prosperous advancement of open data initiatives. This research facilitates the alignment of data providers, users, and regulators, thereby fostering a clear and coherent vision for sustained collaboration and long-term development within the realm of open data.

This study has some directions to continue. First, we can allow for some possibility of mutation in the behavior of the three-party and study the strategic evolution of the three in this case. Next, we can also consider the effect of punishment mechanisms on strategic evolution and introduce an additional category of players, which could be extended to a four-party evolutionary game. Furthermore, the Environments - Classes, Agents, Roles, Groups, and Objects (E-CARGO) model and Role-Based Collaboration (RBC) methodology \cite{zhu2021cargo, zhu2022group, zhu2020computational} have been proposed as a tool for investigating social problems. They are good potential tools for investigating multi-player games. In the future, we may apply E-CARGO and Group Role Assignment with Constraints (GRA+) \cite{zhu2022group} to analyze the problem discussed in this paper. Eventually, we hope that our work will contribute to the relative study of the evolutionary game on open data in the near future.


%

\ifCLASSOPTIONcaptionsoff
  \newpage
\fi

\bibliographystyle{ieeetr}



%

%

\begin{IEEEbiography}[{\includegraphics[width=1in,height=1.25in,clip,keepaspectratio]{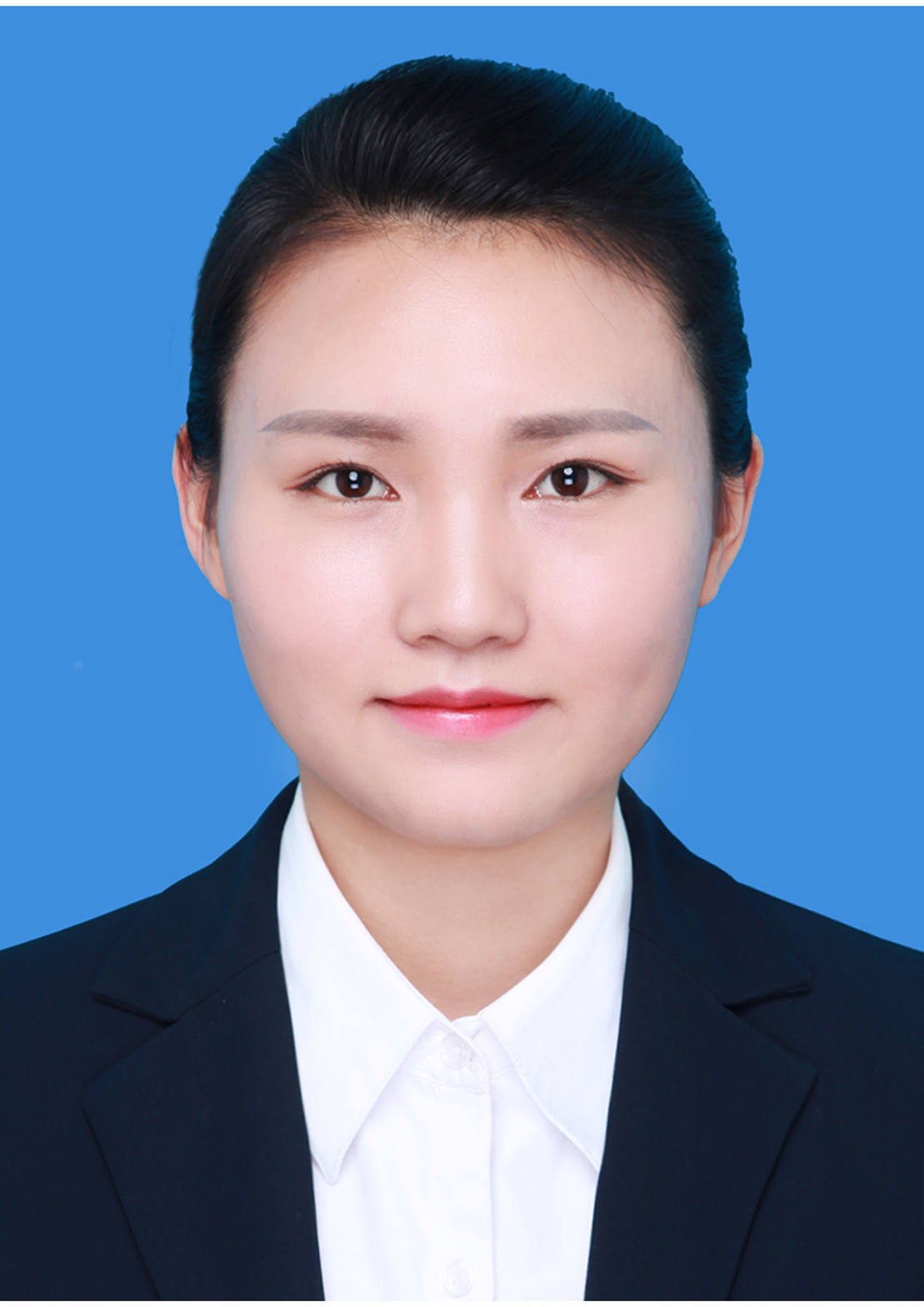}}]{Qin Li} received the M.S. degree in Public Administration from the University of Electronic Science and Technology of China in 2018. She is currently pursuing the Ph.D. degree with School of Public Policy and Administration, Chongqing University, China. Her research interests include social computing, complex networks, and evolutionary games.
\end{IEEEbiography}

\begin{IEEEbiography}[{\includegraphics[width=1in,height=1.25in,clip,keepaspectratio]{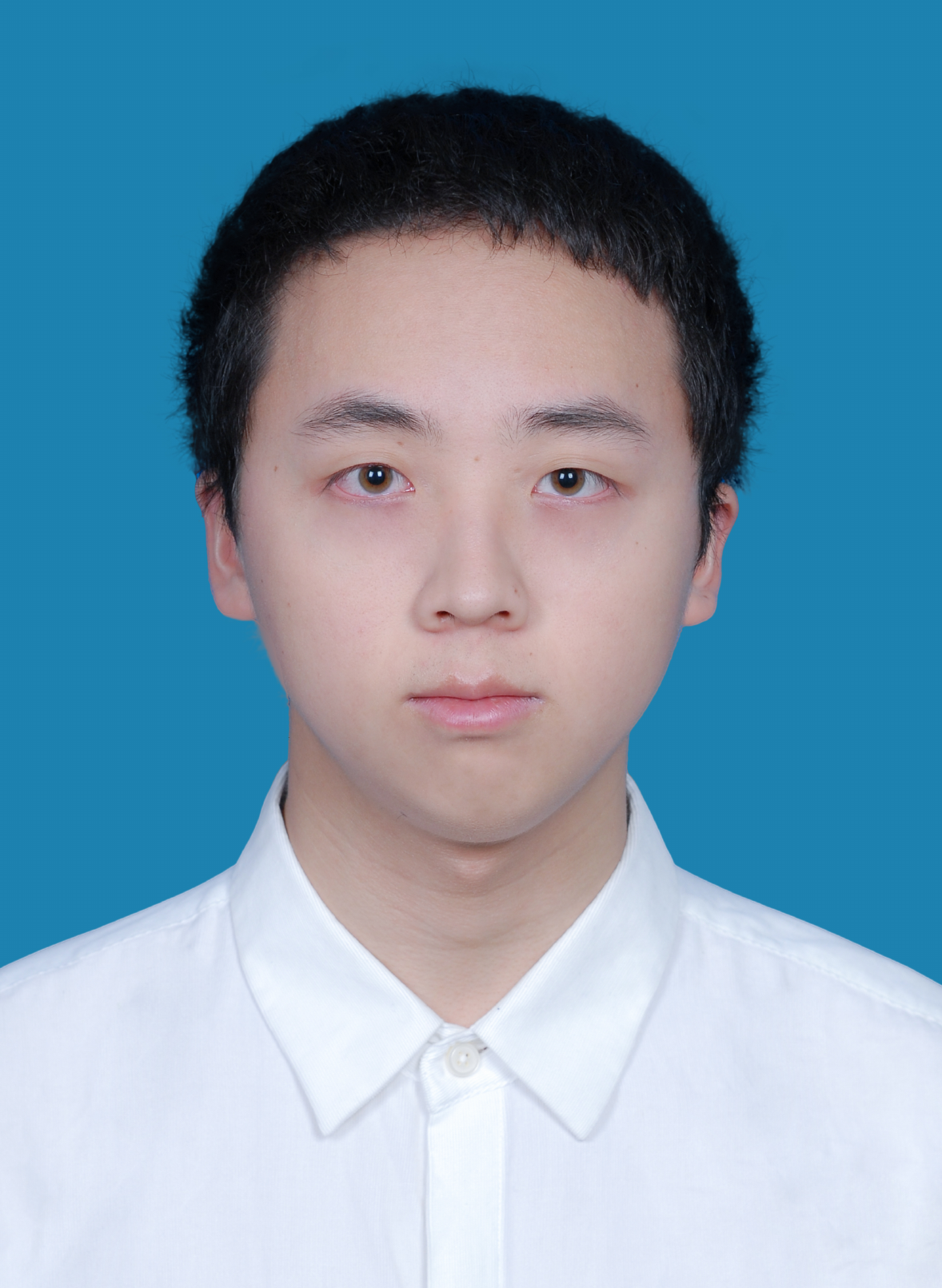}}]{Bin Pi} received the B.E. degree in data science and big data technology from the College of Artificial Intelligence, Southwest University, Chongqing, China, in 2023. He is currently pursuing the M.S. degree in mathematics with the School of Mathematical Sciences, University of Electronic Science and Technology of China, Chengdu, China. His research interests include complex networks, evolutionary games, stochastic processes, and nonlinear science.
\end{IEEEbiography}

\begin{IEEEbiography}[{\includegraphics[width=1in,height=1.25in,clip,keepaspectratio]{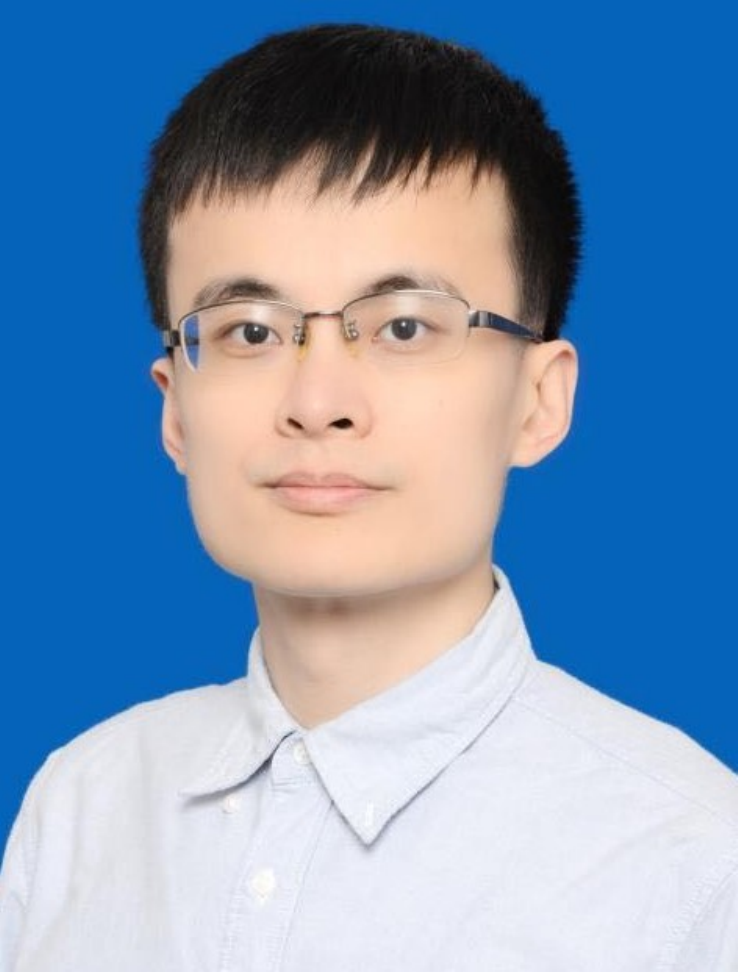}}]{Minyu Feng} (Member IEEE) received the B.S. degree in mathematics from the University of Electronic Science and Technology of China in 2010; the Ph.D. degree in computer science from the University of Electronic Science and Technology of China in 2018. From 2016 to 2017, he was a joint-training Ph.D. student with the Potsdam Institute for Climate Impact Research, Germany, and Humboldt University, Berlin, Germany. Since 2019, he has been an associate professor in the College of Artificial Intelligence, Southwest University, Chongqing, China. He is a Senior Member of China Computer Federation (CCF), an academic member of IEEE and Chinese Association of Automation (CAA). He currently serves as an editorial board member of PLOS ONE, PLOS Complex Systems, and International Journal of Mathematics for Industry, also a guest editor in Entropy and Frontiers in Physics. His research interests include complex systems, stochastic processes, evolutionary games and nonlinear dynamics.
\end{IEEEbiography}

\begin{IEEEbiography}[{\includegraphics[width=1in,height=1.25in,clip,keepaspectratio]{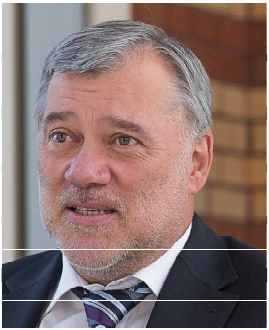}}]{J\"{u}rgen Kurths} received the B.S. degree in mathematics from the University of Rostock, Rostock, Germany, the Ph.D. degree from the Academy of Sciences, German Democratic Republic, Berlin, Germany, in 1983, the Honorary degree from N.I.Lobachevsky State University, Nizhny Novgorod, Russia in 2008, and the Honorary degree from Saratow State University, Saratov, Russia, in 2012.

From 1994 to 2008, he was a Full Professor with the University of Potsdam, Potsdam, Germany. Since 2008, he has been a Professor of nonlinear dynamics with the Humboldt University of Berlin, Berlin, Germany, and the Chair of the Research Domain Complexity Science with the Potsdam Institute for Climate Impact Research, Potsdam, Germany. He has authored more than 700 papers, which are cited more than 60000 times (H-index: 111). His main research interests include synchronization, complex networks, time series analysis, and their applications.

Dr. Kurths was the recipient of the Alexander von Humboldt Research Award from India, in 2005, and from Poland in 2021, the Richardson Medal of the European Geophysical Union in 2013, and the Eight Honorary Doctorates. He is a Highly Cited Researcher in Engineering. He is a member of the Academia 1024 Europaea. He is an Editor-in-Chief of CHAOS and on the Editorial Boards of more than ten journals. He is a Fellow of the American Physical Society, the Royal Society of 1023 Edinburgh, and the Network Science Society.
\end{IEEEbiography}


\vfill


\end{document}